\documentclass[a4paper,reqno,twoside]{amsart}

\usepackage[%
pdfauthor={Habib Ammari,%
            Silvio Barandun,%
            and Alexander Uhlmann},%
pdftitle={Subwavelength Localisation in Disordered Systems},
pdfkeywords={Disordered systems, localisation}
]{hyperref}

\usepackage[left=3.5cm,right=3.5cm,top=3.5cm,bottom=3.5cm,footskip=1.5cm,headsep=1cm,bindingoffset=0.cm,]{geometry}

\usepackage[utf8]{inputenc}
\usepackage[english]{babel}
\usepackage[T1]{fontenc} 
\usepackage{lmodern} 
\rmfamily 
\DeclareFontShape{T1}{lmr}{b}{sc}{<->ssub*cmr/bx/sc}{}
\DeclareFontShape{T1}{lmr}{bx}{sc}{<->ssub*cmr/bx/sc}{}

\numberwithin{equation}{section}
\usepackage{amsmath}
\usepackage{amssymb}
\usepackage{amsthm}
\usepackage{amsfonts}
\usepackage{mathtools}

\usepackage{mathdots}
\usepackage{caption}
\usepackage{subcaption}

\usepackage{dsfont} 
\usepackage[format=hang]{caption}
\usepackage[normalem]{ulem}

\usepackage{standalone}
\usepackage{graphicx}
\usepackage{imakeidx}
\makeindex
\usepackage{setspace}
\usepackage{appendix}
\usepackage{pdfpages}
\usepackage{upgreek}
\usepackage{tabulary}
\usepackage{booktabs}
\usepackage{extarrows}
\usepackage{tabularx}
\usepackage{float}
\restylefloat{table}
\usepackage{enumitem}
\usepackage{csquotes}
\usepackage{bm}
\usepackage{orcidlink}
\usepackage{lipsum}
\usepackage{stackengine}

\usepackage{tikz}
\usepackage{tikz-layers}
\usepackage{tikz-cd}
\usepackage[arrow, matrix, curve]{xy}
\usetikzlibrary{matrix,positioning,decorations.pathreplacing,calc}
\usetikzlibrary{
    decorations.text,%
    decorations.markings,%
    shadows}
\usetikzlibrary{calc,intersections}
\usepackage{adjustbox}
\usepackage{graphicx}

\usepackage[
    style=numeric,
    sorting=nty,
    maxnames=99,
    maxalphanames=5,
    natbib=true,
    backend=bibtex,
    sortcites]{biblatex}

\DeclareNameAlias{default}{family-given}

\AtEveryBibitem{
    \clearfield{url}
    \clearfield{issn}
    \clearfield{isbn}
    \clearfield{urldate}

    \ifentrytype{book}{
        \clearfield{pages}}{%
    }
}

\usepackage{xargs}                      
\usepackage[prependcaption,textsize=tiny,textwidth=3cm]{todonotes}
\newcommandx{\unsure}[2][1=]{\todo[linecolor=red,backgroundcolor=red!25,bordercolor=red,#1]{#2}}
\newcommandx{\change}[2][1=]{\todo[linecolor=blue,backgroundcolor=blue!25,bordercolor=blue,#1]{#2}}
\newcommandx{\info}[2][1=]{\todo[linecolor=OliveGreen,backgroundcolor=OliveGreen!25,bordercolor=OliveGreen,#1]{#2}}
\newcommandx{\improvement}[2][1=]{\todo[linecolor=black,backgroundcolor=black!25,bordercolor=black,#1]{#2}}
\newcommandx{\thiswillnotshow}[2][1=]{\todo[disable,#1]{#2}}

\usepackage[noabbrev,capitalize]{cleveref}
\crefname{proposition}{Proposition}{Propositions}
\crefname{equation}{}{}

\newtheorem{theorem}{Theorem}[section]

\theoremstyle{definition}

\newtheorem{example}[theorem]{Example}

\crefname{assumption}{Assumption}{Assumptions}
\crefname{definition}{Definition}{Definitions}
\crefname{corollary}{Corollary}{Corollaries}
\crefname{enumi}{item}{items}

\creflabelformat{subfigure}{#2\textsc{#1}#3}

\usepackage[final]{microtype}

\makeatletter
\newsavebox\myboxA
\newsavebox\myboxB
\newlength\mylenA

\newcommand*\xoverline[2][0.75]{%
  \sbox{\myboxA}{$\m@th#2$}%
  \setbox\myboxB\null
  \ht\myboxB=\ht\myboxA%
  \dp\myboxB=\dp\myboxA%
  \wd\myboxB=#1\wd\myboxA
  \sbox\myboxB{$\m@th\overline{\copy\myboxB}$}
  \setlength\mylenA{\the\wd\myboxA}
  \addtolength\mylenA{-\the\wd\myboxB}%
  \ifdim\wd\myboxB<\wd\myboxA%
    \rlap{\hskip 0.5\mylenA\usebox\myboxB}{\usebox\myboxA}%
  \else
    \hskip -0.5\mylenA\rlap{\usebox\myboxA}{\hskip 0.5\mylenA\usebox\myboxB}%
  \fi}
\makeatother

\usepackage{fancyhdr}

\fancypagestyle{plain}{%
    \fancyhead[C]{}
    \fancyfoot[C]{}
    
    }
\fancypagestyle{nosection}{%
    \fancyhead[CE]{}
    \fancyhead[CO]{}
    \fancyfoot[CE,CO]{\thepage}
}

\DeclareMathOperator{\Z}{\mathbb{Z}}

\DeclareMathOperator{\R}{\mathbb{R}}
\DeclareMathOperator{\C}{\mathbb{C}}
\DeclareMathOperator{\E}{\mathbb{E}}

\DeclareMathOperator{\tr}{tr}
\renewcommand{\i}{\mathbf{i}}

\renewcommand{\qedsymbol}{$\blacksquare$}

\newcommand{\ds}{\displaystyle}

\DeclareMathOperator{\diag}{diag}
\DeclareMathOperator{\BO}{\mathcal{O}}

\renewcommand{\epsilon}{\varepsilon}
\DeclareMathOperator{\dd}{d\!}

\renewcommand{\i}{\mathbf{i}}

\DeclareMathOperator{\iL}{{\mathsf{L}}}
\DeclareMathOperator{\iR}{{\mathsf{R}}}
\DeclareMathOperator{\iLR}{{\mathsf{L},\mathsf{R}}}

\usepackage{tcolorbox}
\usetikzlibrary{tikzmark,calc}

\newcommand{\mc}[1]{\mathcal{#1}}

\newcommand{\abs}[1]{\left\lvert#1\right\rvert}

\newcommand{\norm}[1]{\left\lVert#1\right\rVert}

\DeclareMathOperator\len{len}

\usepackage{tikz}

\usetikzlibrary{matrix} 
\usetikzlibrary{arrows,arrows.meta,bending} 

\usepackage{nicematrix}

\usetikzlibrary{hobby}

\newcommand{\ldz}{\mathbb{L}_D(\Z)}
\newcommand{\ldm}{\mathbb{L}_D(M)}

\newcommand{\oo}[1]{\chi^{(#1)}}

\usepackage{calc} 

\newcommandx{\silvio}[2][1=]{\todo[linecolor=blue,backgroundcolor=blue!25,bordercolor=blue,#1]{Silvio: #2}}

\newcommandx{\alex}[2][1=]{\todo[linecolor=red,backgroundcolor=red!25,bordercolor=red,#1]{Alex: #2}}

\title[Subwavelength Localisation in Disordered Systems]{Subwavelength  Localisation in Disordered Systems}

\usepackage[left]{lineno}

\begin{document}
 \author[H. Ammari]{Habib Ammari \,\orcidlink{0000-0001-7278-4877}}
\address{\parbox{\linewidth}{Habib Ammari\\
 ETH Z\"urich, Department of Mathematics, Rämistrasse 101, 8092 Z\"urich, Switzerland, \href{http://orcid.org/0000-0001-7278-4877}{orcid.org/0000-0001-7278-4877}}.}
\email{habib.ammari@math.ethz.ch}
\thanks{}

\author[S. Barandun]{Silvio Barandun\,\orcidlink{0000-0003-1499-4352}}
 \address{\parbox{\linewidth}{Silvio Barandun\\
 ETH Z\"urich, Department of Mathematics, Rämistrasse 101, 8092 Z\"urich, Switzerland, \href{http://orcid.org/0000-0003-1499-4352}{orcid.org/0000-0003-1499-4352}}.}
 \email{silvio.barandun@sam.math.ethz.ch}

\author[A. Uhlmann]{Alexander Uhlmann\,\orcidlink{0009-0002-0426-6407}}
 \address{\parbox{\linewidth}{Alexander Uhlmann\\
 ETH Z\"urich, Department of Mathematics, Rämistrasse 101, 8092 Z\"urich, Switzerland, \href{http://orcid.org/0009-0002-0426-6407}{orcid.org/0009-0002-0426-6407}}.}

\email{alexander.uhlmann@sam.math.ethz.ch}

\maketitle

\begin{abstract}
We elucidate the different mechanisms of wave localisation in disordered finite systems of subwavelength resonators, where the disorder is in the spatial arrangement of the resonators.
To do so, we employ the capacitance matrix formalism and develop a variety of tools to understand localisation in this setting. Namely, we adapt the Thouless criterion of localisation to quantify the various mechanisms of localisation. We also employ a propagation matrix approach to prove a Saxon-Hutner type theorem ensuring the existence of band gaps for disordered systems and to predict the corresponding Lyapunov exponents. We further investigate the stabilities of bandgaps and localised eigenmodes with respect to fluctuations in the system by introducing a method of defect mode prediction based on the discrete Green function and illuminating the phenomena of Anderson localisation and level repulsion under global disorder.
\end{abstract}

\bigskip

\par

\bigskip

\\

\section{Introduction}

A subwavelength resonator is a highly contrasted bounded inclusion. The high contrast property gives rise to subwavelength resonances, which are frequencies at which the resonator strongly interacts with incident waves whose wavelengths can be larger by several orders of magnitude \cite{ammari.davies.ea2021Functional}. A typical example in acoustics of a subwavelength resonator is an air bubble in water, where the associated subwavelength resonance is called Minnaert resonance \cite{minnaert1933musical,ammari.fitzpatrick.ea2018Minnaert}. 

Systems of subwavelength resonators can be used to achieve wave localisation at subwavelength scales \cite{review1,review2,review3,sheng,yves.fleury.ea2017Crystalline,lemoult.kaina.ea2016Soda,lemoult.fink.ea2011Acoustic,cummer}. From \cite{feppon.cheng.ea2023Subwavelength,ammari.fitzpatrick.ea2018Minnaert}, we know that their spectral properties are (approximately in terms of the high contrast material parameter) described by the product of two matrices, $V$ and $C$. $V$ is a diagonal matrix where its diagonal terms are the material parameters of each resonator and $C$ is the so-called \emph{capacitance matrix} which contains only geometric information on the system of resonators. In this paper, we focus on subwavelength resonator systems in one dimension. In that case, the corresponding capacitance matrix is tridiagonal; see \eqref{def:C}.

When the system of resonators is periodic (\emph{i.e.},  translationally invariant), the generalised capacitance matrix has, depending on the number of resonators per unit cell, a Toeplitz or block Toeplitz structure perturbed at its edges; see \cite{ammari.davies.ea2021Functional, ammari.barandun.ea2023Edge}.  

The infinitely periodic case is then described by the quasiperiodic capacitance matrix   $C^\alpha$ defined by \eqref{eq:qpcapmat}, where quasiperiodic Bloch solutions and their respective band functions give a complete description of the essential spectrum. For increasing numbers of unit cells, the eigenvalues of the finite system pointwise converge to the essential spectrum of the corresponding Laurent operator; see \cite{ammari.davies.ea2023Spectral}. In other terms, any eigenvalue/eigenvector of $C^\alpha$ can be approximated by eigenvalues/eigenvectors of $C$. 
The existence of a bandgap for $C^\alpha$ shows the existence of a bandgap for $C$ when its size (the number $N$ of resonators) is large enough. Moreover,  the defect eigenmodes in infinite systems introduced by perturbing the material parameters of a fixed number of resonators (\emph{i.e.}, taking $V$ differing from the identity matrix on a fixed number of diagonal entries) have corresponding eigenmodes in finite systems which converge as the size of the system increases; see \cite{ammari.davies.ea2023Convergence}. 

On the one hand, the truncated Floquet--Bloch transform, introduced in \cite{ammari.davies.ea2023Spectral}, gives a concrete way to associate bandgaps and localised eigenmodes to systems of resonators that may be finite periodic or aperiodic (such as  Su--Schrieffer--Heeger (SSH) and dislocated chains).  The main idea is that, when the size of the finite structure is sufficiently large, the structure's eigenmodes are approximately a linear combination of Bloch eigenmodes of the corresponding infinite structure. To compare the discrete eigenvalues of the finite structure to the essential spectrum of the infinite periodic structure, one can reverse engineer the appropriate quasiperiodicities corresponding to these Bloch eigenmodes, taking into account the symmetries in the problem.  The discrete band function and the defect eigenmode calculations in \cite{ammari.davies.ea2023Convergence,ammari.davies.ea2023Spectral} provide a notion of how an eigenmode of the finite problem is approximated either by Bloch eigenmodes (delocalised eigenmodes) or defect modes (localised eigenmodes, corresponding to defect eigenfrequencies inside a bandgap) of the infinite structure. We refer to \cite{TFT} for the mathematical foundations of the Floquet--Bloch transform. 

Furthermore, as shown in \cite{ammari.davies.ea2024Anderson}, the scattering of waves by periodic systems of resonators where their material parameters are perturbed randomly (\emph{i.e.}, entries of $V$ are perturbed randomly) reproduces the characteristic features of Anderson localisation and illustrates strong subwavelength localisation. By Anderson localisation, we mean an increase in both the degree of localisation and the number of localised eigenmodes on average when we increase the random perturbation of the system. Moreover, it was also shown in \cite{ammari.davies.ea2024Anderson} that the hybridisation of subwavelength eigenmodes is responsible for both the repulsion of eigenfrequencies as well as the phase transition, at which point the eigenmode symmetries swap and a very strong localisation is possible.  

All of the recalled results obtained in the subwavelength regime heavily rely on the translation invariance properties of the system of resonators (or analogously the almost Toeplitz or block Toeplitz structure of $C$), together with the use of Floquet--Bloch theory for the analysis of the spectral properties of the corresponding infinite structure.

The aim of this work is to computationally elucidate wave localisation mechanisms in one-dimensional disordered systems. Our work can be seen as a first attempt to construct a consistent picture of how wave localisation occurs at subwavelength scales in disordered systems of subwavelength resonators. By disordered, we mean systems that are not translationally invariant. It is worth emphasising that here the disorder is in the spatial distribution of the resonators. We will show that localised eigenmodes can be created in these systems, under some conditions on the spatial distribution of the resonators, by randomly varying the material parameters of the resonators.  

Our first goal is to introduce the notions of bandgap and midgap eigenvalue for $C$ and to construct several disordered systems that exhibit these phenomena. Then, assuming that $C$ exhibits a bandgap, we show that by randomly perturbing the material parameters of the resonators, as in \cite{ammari.davies.ea2024Anderson}, we can evolve localised eigenmodes from extended eigenmodes and study how the localisation length and the portion of localised eigenmodes in the system vary as the amount of disorder is increased.  

As said above, while for finite periodic structures the associated eigenmodes are approximately linear combinations of Bloch eigenmodes of the corresponding infinite structure, this fact is, in general, not true for disordered systems. However, as first observed by Edwards and Thouless \cite{thoulessnumerical}, if we repeat the disordered structure periodically, then a localised eigenmode must be an eigenmode of the obtained periodic structure for any quasiperiodicity. In other words, by looking at the sensitivity of the eigenmodes to quasiperiodic boundary conditions, we can distinguish between localised and delocalised eigenmodes. While the delocalised eigenmodes of the disordered structure are sensitive to the periodisation of the structure, the localised ones are very insensitive. Conversely, if a Bloch band of the obtained periodic structure is flat (\emph{i.e.}, independent of the quasiperiodicity), then such an eigenfrequency corresponds either to a localised eigenmode or to one of the edges of the bandgap. By this reasoning, we introduce natural notions of localisation and bandgap in disordered structures. The strategy is as follows: by periodising the disordered finite structure and looking at the variation of the eigenfrequencies of the infinitely periodic structure as a function of the quasiperiodicity, we can detect band edges and consequently also the existence of a bandgap in the distribution of the eigenvalues of the capacitance matrix $C$. We may also detect localised eigenmodes inside such bandgaps.

A remarkable finding in this paper is that, depending on the disorder, 
dense regions of eigenvalues outside of the shared pass band may exist. These eigenvalues correspond to \emph{quite} localised eigenmodes (they are less localised than those in the bandgap) 
and accumulate as the number of resonators increases while the number of localised eigenmodes remains finite. Such an accumulation does not exist in periodic systems of subwavelength resonators. They also do not occur in dislocated chains of resonators or finite chains of the SSH type. In both cases, the structure is composed of two half-periodic systems having the same band structure. In some sense, to display these 
eigenmodes, disordered systems must differ significantly from periodic ones and break even local translation invariance. These eigenmodes, called \emph{hybridised bound eigenmodes}, can then occur for resonant frequencies that are partially supported in these systems, whereas defect modes that are supported nowhere are fully localised.

Finally, it is worth mentioning that our results in this paper generalise those obtained on the localisation of electrons in disordered lattices \cite{thouless1974Electrons,minami,molcanov1981local,mirani2} to classical systems of subwavelength resonators where, despite being in the low-frequency regime, the subwavelength resonances of the system interact strongly with the disorder. For similar results on the continuous Schr\"odinger model, we refer the reader to \cite{hislop} and the references therein. It should be noticed that the localisation mechanism for the continuous Schr\"odinger model, which is due to potential wells, is fundamentally different from the discrete models studied here. To our knowledge, the landscape localisation theory \cite{filoche,arnold.filoche.ea2022Landscape,lyra.mayboroda.ea2015Dual} is one of the most intriguing ways to detect localised eigenfunctions in the continuous Schr\"odinger model. An effective potential (known as a landscape function) finely governs the confinement of the localised eigenfunctions. In this picture, the boundaries of the localisation subregions for low-energy eigenfunctions correspond to the barriers of this effective potential, and the long-range exponential decay characteristic of localisation is explained as the consequence of multiple tunnelling in the dense network of barriers created by this effective potential. In \cite{bellis}, the localization landscape theory
is extended to classical scalar waves that exhibit high-frequency localisation.
In \cite{davies.lou2023Landscape}, a landscape theory for the generalised capacitance matrix is developed and used to predict wave localisation positions in random and aperiodic systems of subwavelength resonators. In \cite{tao}, it is shown that the landscape localisation theory is valid for all $M$-matrices. 

The paper is organised as follows. In \cref{sec:blockdisorder}, we first recall the capacitance matrix formulation of the resonance problem for a subwavelength resonator system. Then we introduce the notation and the construction associated with block disordered systems. In particular, we give three examples of disordered systems (SSH, disordered, and randomly sampled blocks), which can be understood in the block disordered framework. 
In \cref{sec: band gaps}, we introduce the Thouless criterion of localisation for the subwavelength setting by imposing quasiperiodic boundary conditions on the systems. 
Using this tool we investigate the localisation behaviour in disordered systems and uncover two distinct mechanisms of localisation: \emph{Bandgap localisation}, which stems from the bandgaps of the constituent parts and \emph{disordered localisation}, which occurs in very large disordered systems and is significantly weaker than the bandgap one. Both the bandgap localisation and disordered localisation can be tweaked, respectively, by either varying the bandgaps of the constituent blocks or their relative sampling densities. In \cref{sec:propagation}, we show how to predict bandgaps in block disordered systems by analysing their propagation matrices and use this framework to understand the localisation mechanisms in the different spectral regions. In particular, we investigate the mechanisms responsible for the occurrence of hybridised bound eigenmodes. We also relate the Lyapunov exponent of the total disordered system to those of its constituents to predict the decay of eigenmodes.
In \cref{sec:defects}, by randomly perturbing the material parameters of the resonators, we reproduce Anderson localisation features in disordered systems, which can be understood by considering the simple phenomena of \emph{hybridisation} and \emph{level repulsion}. In \cref{sec4:AL}, we investigate the stability and behaviour of the bandgaps and hybridised bound eigenfrequencies under random perturbations of the material parameters of all resonators. In \cref{sec:prediction}, we compare
the performance of the landscape function for the detection of localised defect modes in perturbed periodic systems with that in disordered systems where strongly localised bound eigenmodes occur. We also show that the behaviour of the fractal dimension related the eigenmodes of randomly perturbed disordered systems can be used to identify the mobility edges even in the presence of hybridised eigenmodes.
In \cref{sec:conclusion}, we summarise our main contributions in this paper and formulate some open questions. Finally, in \cref{appendixA}, we apply the methodology developed in this paper to arbitrary banded Hermitian $M$-matrices to detect their localised eigenvectors. 

\section{Setting}\label{sec:blockdisorder}
In this section, our aim is to provide an overview of the \emph{subwavelength} setting as well as to introduce the notation and construction associated with \emph{block disordered} systems of such subwavelength resonators.
\subsection{Subwavelength resonators}\label{sec:subwavelength setting}
The central systems of interest in this work are one-dimensional chains of subwavelength resonators (see \cite{ammari.davies.ea2021Functional,feppon.cheng.ea2023Subwavelength, cbms}). That is, we consider an array $\mc D=\bigsqcup_{i=1}^N (x_i^{\iL}, x_i^{\iR})$ consisting of $N$ resonators $D_i = (x_i^{\iL}, x_i^{\iR})$. We will denote $\ell_i = x_i^{\iR} - x_i^{\iL}$ for $1\leq i \leq N$ the \emph{length} of the $i$\textsuperscript{th} resonator and $s_i = x_{i+1}^{\iL}-x_{i}^{\iR}$ for $1\leq i \leq N-1$ the \emph{spacing} between the $i$\textsuperscript{th} and $(i+1)$\textsuperscript{th} resonator. 

The resonators are distinct from the background medium due to differing wave speeds and densities, given by
\begin{equation}
    v(x) = \begin{cases}
        v_i &\text{if } x\in D_i,\\
        v &\text{if } x\in \R\setminus \mc D, 
    \end{cases} \quad
    \rho(x) = \begin{cases}
        \rho_b &\text{if } x\in D_i,\\
        \rho &\text{if } x\in \R\setminus \mc D.
    \end{cases}
\end{equation}
After further imposing an outward radiation condition, we obtain the following coupled system of Helmholtz equations for the resonant modes (see \cite[(1.6)]{feppon.ammari2022Subwavelength}). The resonance problem is to find $\omega$ such that there is a nontrivial solution $u$ to

\begin{align}
\label{waveeq}
    \begin{dcases}
\ds \frac{\mathrm{d}^2}{\mathrm{d}x^2} u + \frac{\omega^2}{v^2} u = 0,  &\text{in }  
\R\setminus \mc D,\\
\ds \frac{\mathrm{d}^2}{\mathrm{d}x^2} u + \frac{\omega^2}{v_i^2} u = 0,  & \text{in }   D_i, i=1,\ldots, N,\\
 u\vert_{\iR}(x^{\iLR}_i) - u\vert_{\iL}(x^{\iLR}_i) = 0,                                                                 & \text{for } i=1, \ldots, N-1,\\
        \left.\frac{\dd u}{\dd x}\right\vert_{\iR}(x^{\iL}_i) - \frac{\rho_b}{\rho} \left.\frac{\dd u}{\dd x}\right\vert_{\iL}(x^{\iL}_i) = 0, &   \text{for } i=1, \ldots, N-1,          \\
        \left.\frac{\dd u}{\dd x}\right\vert_{\iL}(x^{\iR}_i) - \frac{\rho_b}{\rho} \left.\frac{\dd u}{\dd x}\right\vert_{\iR}(x^{\iR}_i) = 0, &   \text{for } i=1, \ldots, N-1,          \\
\bigg( \frac{\mathrm{d}}{\mathrm{d} |x|} - \i \frac{\omega}{v} \bigg) u(x) =0, & \text{for } |x| \text{ large enough,} 
\end{dcases}
\end{align}
where for a function $w$ we denote by 
\begin{align*}
    w\vert_{\iL}(x) \coloneqq \lim_{s \downarrow 0} w(x-s) \quad \mbox{and} \quad  w\vert_{\iR}(x) \coloneqq \lim_{s \downarrow 0} w(x+s)
\end{align*}
if the limits exist.

We are interested in \emph{subwavelength high-contrast regime}. Namely, we denote by $\delta = \rho_b / \rho$ the \emph{contrast parameter} and consider the resonant frequencies  $\omega(\delta)$ that satisfy
\begin{equation*}
    \omega(\delta) \to 0 \quad \text{ as } \quad \delta \to 0.
\end{equation*}

As shown in \cite{feppon.cheng.ea2023Subwavelength}, in this regime there exist exactly $N$ subwavelength resonant frequencies, characterised in leading order by a \emph{material matrix} $V$ that is diagonal and a tridiagonal \emph{capacitance matrix} $C$, respectively defined as
\begin{gather} 
    V = \diag \left(\frac{v_1^2}{\ell_1}, \dots, \frac{v_N^2}{\ell_N}\right) \in \R^{N\times N}, \label{def:v}
 \end{gather}
 and    
\begin{gather}
C = \left(\begin{array}{cccccc}
     \frac{1}{s_1}& -\frac{1}{s_1} \\
     -\frac{1}{s_1}& \frac{1}{s_1}+\frac{1}{s_2}& -\frac{1}{s_2} \\
     & -\frac{1}{s_2} & \frac{1}{s_2}+\frac{1}{s_3}& -\frac{1}{s_3}\\
     &&\ddots&\ddots&\ddots \\
     &&&-\frac{1}{s_{N-2}}& \frac{1}{s_{N-2}}+\frac{1}{s_{N-1}}& -\frac{1}{s_{N-1}}\\
     &&&&-\frac{1}{s_{N-1}}&\frac{1}{s_{N-1}}
\end{array}\right) \in \R^{N\times N}. \label{def:C}
\end{gather}
The following results are from \cite{feppon.cheng.ea2023Subwavelength}. 
\begin{theorem}\label{thm:capapprox}
    Consider a system consisting of $N$ subwavelength resonators in $\R$. Then, there exist exactly $N$ subwavelength resonant frequencies $\omega(\delta)$ that satisfy $\omega(\delta)\to 0$ as $\delta\to 0$. Moreover, the $N$ resonant frequencies are approximately given by 
    \[
        \omega_i(\delta) = \sqrt{\delta\lambda_i} + \mc O (\delta),
    \]
    where $(\lambda_i)_{1\leq i\leq N}$ are the eigenvalues of the eigenvalue problem
    \[
        VC\bm u_i = \lambda_i \bm u_i.
    \]
    Furthermore, the corresponding resonant modes $u_i(x)$ are approximately given by 
    \[
        u_i(x) = \sum_{j=1}^N \bm u_i^{(j)}V_j(x) + \mc O (\delta),
    \]
    where $\bm u_i^{(j)}$ is the $j$\textsuperscript{th} entry of the vector $\bm u_i$ and $V_j(x)$, $j=1,\ldots, N$, are defined by 
    \begin{align*}
        \begin{dcases}
          -\frac{\dd{^2}}{\dd x^2} V_j(x) =0, & x\in\R\setminus \mc D, \\
          V_j(x)=\delta_{ij},              & x\in D_i,  i=1,\ldots, N,                         \\
          V_j(x) = \BO(1),                  & \mbox{as } \abs{x}\to\infty.
        \end{dcases}
        \label{eq: def V_i}
    \end{align*}
\end{theorem}

We can thus fully understand the subwavelength resonant modes by studying the eigenvalue problem for the \emph{generalised capacitance matrix} $\mc C \coloneqq VC$ and will often use $\lambda_i$ and $\omega_i$ interchangeably, where it is clear from context.

\subsection{Block disordered systems}
As we are interested in the band structure of disordered systems, where local translation invariance is broken, we must develop ways to construct and describe such systems. One such way is to consider a finite number of distinct \enquote{building blocks} consisting of (possibly multiple) subwavelength resonators. Later, constructing disordered resonator arrays from simple building blocks will enable us to characterise the properties of the array by looking at the building blocks.

A subwavelength block disordered system is a chain of subwavelength resonators consisting of $M$ blocks of resonators $B_{\oo{j}}$ sampled accordingly to a sequence $\chi \in \{1,\dots , D\}^M \eqqcolon \ldm$ from a selection $B_1, \dots, B_D$ of $D$ distinct resonator blocks, arranged on a line. Every resonator block $B_j$ is characterised by a sequence of tuples $(v_1,\ell_1,s_1), \dots, (v_{\len(B_j)},\ell_{\len(B_j)},s_{\len(B_j)})$ that denote the wave speed, length, and spacing of each constituent resonator. Here, $\len(B_j)$ denotes the total amount of resonators contained within the block $B_j$.

We will often abuse notation and write $\mc D = \bigcup_{j=1}^M B_{\oo{j}}$ to denote the resonator array constructed by sampling the blocks $B_1, \dots B_D$ in accordingly to $\chi\in \ldm$ and then arranging them in a line. Having thus constructed an array of subwavelength resonators, we can alternatively write that $\mc D = \bigcup_{i=1}^N D_i$ in line with \cref{sec:subwavelength setting}.
Note that as $M$ denotes the total number of sampled blocks and $N$ the total number of resonators, we always have $M\leq N$.

\begin{example}\label{ex:standard_blocks}
    A simple but nontrivial disordered system is obtained by sampling from a set of two blocks $B_1$ and $B_2$.
    $B_1=B_{single}$ denotes a single resonator block with $\len(B_{single}) = 1$ and $\ell_1(B_{single}) = s_1(B_{single})  = 2$ while $B_2=B_{dimer}$ is a dimer resonator block with $\len(B_{dimer}) = 2$ and $\ell_1(B_{dimer}) = \ell_2(B_{dimer})=1$ and $s_1(B_{dimer})=1, s_2(B_{dimer})=2$. We choose all wave speeds to be equal to $1$.
    An example of a single realisation  of this system with $M=14$ is described by the following sequence:
    \[
        \chi = (1,1,1,1,2,1,1,1,1,1,2,1,1,1) \in \{1,2\}^M . 
    \]
    In \cref{fig:disorderedsketch}, another realisation corresponding to the sequence $\chi = (1,2,1)$ is illustrated.
\end{example}

Notably, when repeated periodically, both the single resonator and the dimer share a lower band. However, as will be shown in \Cref{sec:propagation}, the dimer possesses an additional upper band and the distinct properties of these building blocks affect the resonant frequencies of the resulting array. 

\begin{figure}
    \centering
    \begin{tikzpicture}[scale=1.0, every node/.style={font=\footnotesize}]

    \draw[thick] (0,0) -- (2,0);
    \node[below] at (1,0) {$D_1$};
    
    \draw[
    decorate,
    decoration={brace, mirror, amplitude=5pt}
    ] (0,-0.35) -- (4,-0.4)
    node[midway, yshift=-0.45cm] {$B_1$};
    
    \draw[-, dotted] (0,0) -- (0,0.7);
    \draw[-, dotted] (2,0) -- (2,0.7);
    \draw[|-|, dashed] (0,0.7) -- node[above]{$\ell_1(B_1)$} (2,0.7);
    
    \draw[-, dotted] (4,0) -- (4,0.7);
    \draw[|-|, dashed] (2,0.7) -- node[above]{$s_1(B_1)$} (4,0.7);
    
    \draw[thick] (4,0) -- (5,0);
    \node[below] at (4.5,0) {$D_2$};
    \draw[thick] (6,0) -- (7,0);
    \node[below] at (6.5,0) {$D_3$};
    
    \draw[
    decorate,
    decoration={brace, mirror, amplitude=5pt}
    ] (4,-0.4) -- (9,-0.4)
    node[midway, yshift=-0.45cm] {$B_2$};
    
    \draw[-, dotted] (5,0) -- (5,0.7);
    \draw[-, dotted] (6,0) -- (6,0.7);
    \draw[-, dotted] (7,0) -- (7,0.7);
    \draw[-, dotted] (9,0) -- (9,0.7);
    \draw[|-|, dashed] (4,0.7) -- node[above]{$\ell_1(B_2)$} (5,0.7);
    \draw[|-|, dashed] (5,0.7) -- node[above]{$s_1(B_2)$} (6,0.7);
    \draw[|-|, dashed] (6,0.7) -- node[above]{$\ell_2(B_2)$} (7,0.7);
    \draw[|-|, dashed] (7,0.7) -- node[above]{$s_2(B_2)$} (9,0.7);

    \begin{scope}[shift={(9,0)}]
    \draw[thick] (0,0) -- (2,0);
    \node[below] at (1,0) {$D_4$};
    \draw[
    decorate,
    decoration={brace, mirror, amplitude=5pt}
    ] (0,-0.4) -- (4,-0.4)
    node[midway, yshift=-0.45cm] {$B_1$};
    
    \draw[-, dotted] (0,0) -- (0,0.7);
    \draw[-, dotted] (2,0) -- (2,0.7);
    \draw[|-|, dashed] (0,0.7) -- node[above]{$\ell_1(B_1)$} (2,0.7);
    
    \draw[-, dotted] (4,0) -- (4,0.7);
    \draw[|-|, dashed] (2,0.7) -- node[above]{$s_1(B_1)$} (4,0.7);
    \end{scope}
\end{tikzpicture}
    \caption{A block disordered system consisting of two single resonator blocks $B_1$ and a dimer block $B_2$ arranged in a chain given by the sequence $\chi= (1,2,1)$. It thus consists of $M=3$ blocks and $N=4$ resonators $D_1,\dots ,D_4$ in total.}
    \label{fig:disorderedsketch}
\end{figure}

\subsection{Arrangements}\label{sec:blocksystems}
We consider three families of disordered systems, all of which can be understood in the block disordered framework:
\begin{description}
    \item[SSH] Inspired by the SSH model and its quantum-mechanical analogue, this model is composed of two periodic systems with different unit cells joined together. In \cite{ammari.barandun.ea2024Exponentially}, this system is studied for systems of subwavelength resonators;
    \item[Dislocated] This model is composed of an array of dimers, one of which gets \emph{dislocated} by increasing the intra-resonator distance. We point to \cite{ammari.davies.ea2022Robust} for a study of this model;
    \item[Randomly sampled blocks] This model consists of two distinct resonator blocks $B_1$ and $B_2$ sampled $M$-times independently and identically with probability $p_1$ and $p_2=1-p_1$. Although we will usually consider $B_1$ and $B_2$ as in \cref{ex:standard_blocks}, we will also sometimes consider two distinct dimer blocks, partially sharing the upper pass band (see \cref{fig: var and eigs prop matrix}(\textsc{b})).
\end{description}
The SSH and dislocated systems can be understood as block disordered systems consisting of a dimer block $B_2$ (we will select the same dimer block as in \cref{ex:standard_blocks}) and some defect block $B_1$. Their corresponding block sequences are respectively given by $\chi = (\dots, 2,2,1,2,2, \dots)\in \ldm$. 

The family of systems consisting of randomly sampled disordered blocks is more strongly disordered than both the SSH and dislocated systems, as even local translation invariance is violated. The eigenmodes are thus no longer given (at least locally) by Bloch eigenmodes, making it impossible to obtain any classical dispersion relation. Consequently, numerical methods such as the truncated Floquet--Bloch transform, which allows the recovery of the band structure for the SSH and dislocated systems \cite{TFT}, fail to recover anything in the disordered block case.

However, while strongly disordered systems might no longer exhibit the full band structure of their periodic counterparts, we can still observe the familiar notions of band regions and bandgaps. Furthermore, by controlling the sampling densities $p_1$ and $p_2$, we can precisely control the amount of disorder introduced into the system. It is worth emphasising that a crucial consequence of this random block sampling is that, as the number of sampled blocks $M$ goes to infinity, any local arrangement of blocks must occur infinitely often almost surely.
 
\section{Bandgaps in disordered systems and the Thouless criterion} \label{sec: band gaps}
It is well known that the subwavelength part of the spectrum of periodic systems of subwavelength resonators where there are $N$ resonators per unit cell is given by the union of the first $N$ band functions as a result of  Floquet--Bloch's theorem \cite{ammari.davies.ea2021Functional,cbms,jde}. Recent results obtained in \cite{ammari.davies.ea2023Spectral,TFT} have shown that 
the subwavelength part of the spectrum of the infinite system can be approximated by the subwavelength eigenfrequencies associated with large size, finite systems. Furthermore, the boundary of the essential spectrum gives the mobility edges discriminating between localised and delocalised eigenmodes.

This theory does not apply to disordered systems, but the distinction and prediction of whether a frequency corresponds to a localised or a delocalised eigenmode are still highly relevant. To do that, we generalise the idea introduced in \cite{thoulessnumerical} where the sensitivity of the eigenvalues to the choice of periodic or antiperiodic boundary conditions was used as a criterion for localisation, as follows. The main difference between a localised and a delocalised eigenmode is its magnitude at the edges of the system: While a localised eigenmode is \emph{exponentially} small at the edges, a delocalised eigenmode is typically of the order $\mc O (1/\sqrt{N})$ there. Consequently, localised eigenmodes are not sensitive to quasiperiodic boundary conditions, while delocalised eigenmodes are. It is therefore natural to consider the system with quasiperiodic boundary conditions at its edges. 

\subsection{Quasiperiodic boundary conditions and the Thouless criterion}

In line with the intuition developed above, given some array of $N$ subwavelength resonators described by a generalised capacitance matrix $\mc C\in \R^{N\times N}$, we aim to measure the sensitivity of the $N$ eigenfrequencies $\lambda_1<\dots <\lambda_N$\footnote{Note that $\mc C$ is diagonalisable because it is similar to a Hermitian matrix, and has only simple eigenvalues because it is tridiagonal.} to quasiperiodic boundary conditions. To that end, we take the finite resonator array and periodise it by copying it infinitely often in both directions. As a consequence, we obtain an infinitely periodic array consisting of a unit cell containing $N$ resonators. Using Floquet-Bloch theory, as for example in \cite{ammari.barandun.ea2023Edge}, we find that the spectrum of the periodised system is given by $N$ bands $\lambda_1(\alpha),\dots, \lambda_N(\alpha)$ for some quasimomentum $\alpha$ in the first Brillouin zone $Y^* \coloneqq[-\pi/\iL,\pi/\iL]$. Here, $\iL$ denotes the length of the unit cell, which, by construction, corresponds to the length of the original resonator array.

Following \cite{ammari.barandun.ea2023Edge}, we can consider the \emph{quasiperiodic capacitance matrix} 
\begin{gather}\label{eq:qpcapmat}
    C^\alpha = \left(\begin{array}{cccccc}
         \frac{1}{s_N}+\frac{1}{s_1}& -\frac{1}{s_1}&&&& - \frac{e^{-\i\alpha\iL}}{s_N} \\
         -\frac{1}{s_1}& \frac{1}{s_1}+\frac{1}{s_2}& -\frac{1}{s_2} \\
         & -\frac{1}{s_2} & \frac{1}{s_2}+\frac{1}{s_3}& -\frac{1}{s_3}\\
         &&\ddots&\ddots&\ddots \\
         &&&-\frac{1}{s_{N-2}}& \frac{1}{s_{N-2}}+\frac{1}{s_{N-1}}& -\frac{1}{s_{N-1}}\\
         - \frac{e^{\i\alpha\iL}}{s_N}&&&&-\frac{1}{s_{N-1}}&\frac{1}{s_{N-1}}+\frac{1}{s_N}
    \end{array}\right) \in \mathbb{C}^{N\times N},
\end{gather}
and find the bands functions at $\alpha$ to be the eigenvalues of the \emph{generalised quasiperiodic capacitance matrix} $\mc C^\alpha \coloneqq VC^\alpha$. We note that $\mc C^\alpha$ is again similar to a Hermitian matrix and has only one-dimensional eigenspaces, for any $\alpha\in Y^*$.
As a consequence, the band functions cannot intersect and the ordering $\lambda_1(\alpha)<\dots < \lambda_N(\alpha)$ holds for any $\alpha\in Y^*$ allowing us to uniquely associate each band function $\alpha\mapsto \lambda_i(\alpha)$ with the corresponding eigenvalue $\lambda_i$ of the finite system, consisting of $N$ resonators.

However, there is one slight subtlety in imposing these quasiperiodic boundary conditions. Namely, as can be seen in \cref{eq:qpcapmat}, we need to choose the spacing $s_N$ corresponding to the separation with which the unit cells are arranged during periodisation. Although choosing an arbitrary $s_N$ that is roughly of order $s_1,\dots s_{N-1}$ does not introduce a too large error (especially if $N$ is large), we can do better in the block disordered case.
When arranging the blocks $B_1, \dots B_D$ in a linear array according to some sequence $\chi \in \ldm$, the spacing of the last resonator of the final block is effectively discarded because there are no more resonators to its right anyway. However, when periodising, this is the natural candidate for $s_N$, ensuring that the periodised array also contains only blocks in $B_1, \dots, B_N$. In particular, we choose
\begin{equation}\label{eq:sn}
    s_N \coloneqq s_{\len (B_{\oo{M}})} (B_{\oo{M}}).
\end{equation}

Having discussed how to find the bands, we can now estimate the \emph{level shifts} defined as the average shift in the band function frequency
\begin{equation}
    \delta\lambda_i \coloneqq \frac{1}{2\pi\iL}\int_{\alpha\in Y^*} \abs{\lambda_i(\alpha)-\lambda_i(0)}d\alpha .
\end{equation}

In practice, because of non-degeneracy of the band functions (as illustrated in \cref{fig:band_function_qpbc}), it is often sufficient to only consider periodic and anti-periodic boundary conditions (\emph{i.e.}, $\alpha\in \{0,\pi/\iL\}$) to capture the scale of the level shift
\begin{equation}\label{eq:energy_shift_approx}
    \delta\lambda_i \approx \abs{\lambda_i(\pi/\iL)-\lambda_i(0)}.
\end{equation}

The \emph{Thouless ratio} $g(\lambda_i)$ for a given frequency $\lambda_i$ is then defined as the ratio of the level shift over the \emph{mean level spacing} $\Delta(\lambda_i)$ around that frequency. Namely,
\begin{equation}\label{eq:thoulessratio}
    g(\lambda_i) \coloneqq \frac{\delta \lambda_i}{\Delta(\lambda_i)} \quad \text{with} \quad \Delta(\lambda_i) = \frac{1}{D(\lambda_i)\iL},
\end{equation}
{
where $D(\lambda)$ denotes the \emph{estimated density of states}\footnote{Intuitively, the density of states at $\lambda$ is the \emph{expected} amount of eigenvalues in the interval $[\lambda - \frac{\mathrm{d}\lambda}{2}, \lambda + \frac{\mathrm{d}\lambda}{2}]$.} at $\lambda$ obtained by a Gaussian kernel density estimate from the empirical eigenvalues $\lambda_1, \dots, \lambda_N$. That is, 
\begin{equation}
 D(\lambda) \coloneqq \frac{1}{\sigma h N\sqrt{2\pi}}\sum_{i=1}^N \exp\left(-\frac{(\lambda-\lambda_i)^2}{2h^2\sigma^2}\right),
\end{equation}
where $\sigma$ denotes the empirical standard deviation of the eigenvalues $\lambda_1, \dots, \lambda_N$ and $h>0$ denotes the \emph{bandwidth} and acts as a smoothing parameter. In practice, we choose $h=10^{-2}$.}

Intuitively, the Thouless ratio for a given eigenfrequency $\lambda_i$ as calculated in \cref{eq:thoulessratio} measures the sensitivity of the resonant frequencies to boundary conditions, normalised by the amount of eigenfrequencies close by. This normalisation makes sense because, as noted above, the band functions cannot cross and are thus \enquote{sandwiched} between the surrounding bands, limiting their variation proportionally.

Following \cite{thoulessnumerical, thouless1974Electrons}, our aim is to relate the Thouless ratio $g(\lambda_i)$ of the eigenfrequencies $\lambda_i$ of $\mc C$ to the localisation behaviour of the associated eigenmodes $\bm u_i$. In particular, we expect that
\begin{equation}\label{eq:thoulesscriterion}
    \bm u_i \text{ is } \begin{cases}
        \text{delocalised} & \text{ if } g(\lambda_i) \approx 1,\\
         \text{localised} & \text{ if } g(\lambda_i) \ll 1.
    \end{cases}
\end{equation}
This constitutes the \emph{Thouless criterion of localisation} applied to the disordered systems of subwavelength resonators.

Note that because in one dimension the band functions cannot intersect, we can never achieve $\delta\lambda_i \gg \Delta(\lambda_i)$ (or equivalently $g(\lambda_i)\gg 1$). This makes \cref{eq:thoulesscriterion} a complete characterisation.

\subsection{Thouless criterion and band structure for block disordered systems}
\begin{figure}[h]
    \centering
    \includegraphics[width=\textwidth]{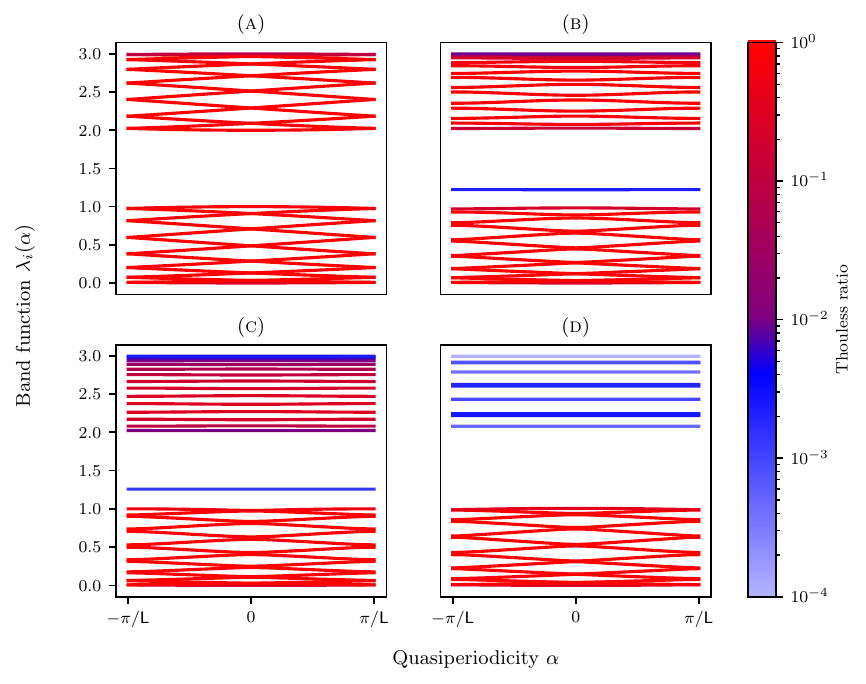}
    \caption{The band functions when imposing quasiperiodic boundary conditions to a (\textsc{a}) periodic, (\textsc{b}) SSH, (\textsc{c}) dislocated, and  (\textsc{d}) random block system. For each band, the colour indicates its Thouless criterion for that band. (\textsc{a}) consists of the dimer block $B_1$ and (\textsc{b-d}) are constructed as in \cref{sec:blocksystems}.
    The interval $(1,2)$ seems to act as a bandgap for all systems. The systems (\textsc{b}) and (\textsc{c}) have a defect midgap eigenmode (in light blue) while system (\textsc{d}) has an accumulation of near-to-zero-variance bands.
    }
    \label{fig:band_function_qpbc}
\end{figure}
To illuminate this connection, we evaluate the Thouless criterion for the block disordered settings introduced in the previous section. 
In \cref{fig:band_function_qpbc}, we show that the Bloch band functions for a (\textsc{a}) periodic, (\textsc{b}) SSH, (\textsc{c}) dislocated and (\textsc{d}) disordered block structure, obtained by imposing quasiperiodic boundary conditions on the respective finite systems, as discussed in the previous subsection. 

In \cref{fig:band_function_qpbc}, we plot the band functions and the Thouless ratio of these systems and observe the following:
\begin{enumerate}[label=(\roman*)]
    \item In all cases, there appears to be a bandgap about the interval $(1,2)$ surrounded by bands $(0,1)$ and $(2,3)$. This reflects the band structure of the dimer building block $B_1$ as in \cref{ex:standard_blocks}; 
    \item For the two defected block disordered systems (\textsc{b}-\textsc{c}) we can observe a band belonging to the defect mode lying in the bandgap. This band is essentially flat, which reflects the fact that the corresponding defect mode is localised and hence insensitive to boundary conditions;
    \item As the band functions get closer to the bandgap, they become more densely spaced and their variation decreases as a result. This observation supports the introduction of the Thouless criterion $g(\lambda_i)$ as the band variation, normalised by the mean level spacing $\Delta(\lambda_i)$ as a measure of the sensitivity to boundary conditions;
    \item The band functions are symmetric about $\alpha=0$ and monotonic in $[0,\pi/\iL]$. This supports the conclusion that \cref{eq:energy_shift_approx} is a valid approximation for the level shift $\delta\lambda_i$. 
    \item Finally, in (\textsc{d}) we observe that flat bands need not be isolated and may appear in dense regions.
\end{enumerate}

Physically, a periodised structure cannot have a perfectly flat Bloch band function, as Floquet--Bloch's theorem would then imply a non-empty point spectrum, which is impossible because of the translation invariance property of the system. Nevertheless, \cref{fig:band_function_qpbc}(\textsc{b-c}) shows that numerically flat bands are associated with localised eigenmodes in SSH and dislocated systems.

\subsection{Characterisation of the localisation}
A natural question that arises at this point is whether or not a low-band function variation also indicates a high localisation.  To determine the localisation of the eigenvectors $\bm u_i \in \C^N$ of $\mathcal{C}$, we calculate their \emph{inverse participation ratios} defined by $\norm{\bm u_i}_4/\norm{\bm u_i}_2$, which take values between $1/\sqrt[4]{N}$ (for completely delocalised eigenvectors) and $1$ (for completely localised ones). In \cref{fig:var_vs_loc}, for all four cases presented in \cref{fig:band_function_qpbc}, we calculate their respective eigenvalues $\lambda_i$ and compare the inverse participation ratios of the corresponding eigenvectors $\bm u_i$ with the Thouless ratio $g(\lambda_i)$ obtained by periodisation. We can see excellent agreement between the two measures as the Thouless ratio is able to accurately identify the localised eigenvectors\footnote{With one limitation: Towards the band edges the Gaussian kernel density estimate tends to underestimate the density of states $D(\lambda)$ due to its discontinuity at the band edge, leading to an erroneously low Thouless ratio. However, this occurs only for a small and fixed amount of eigenvalues and therefore, as the system size $N$ increases, the fraction of affected eigenvalues goes to $0$.}. 

\begin{figure}[h]
    \centering
    \includegraphics[width=\textwidth]{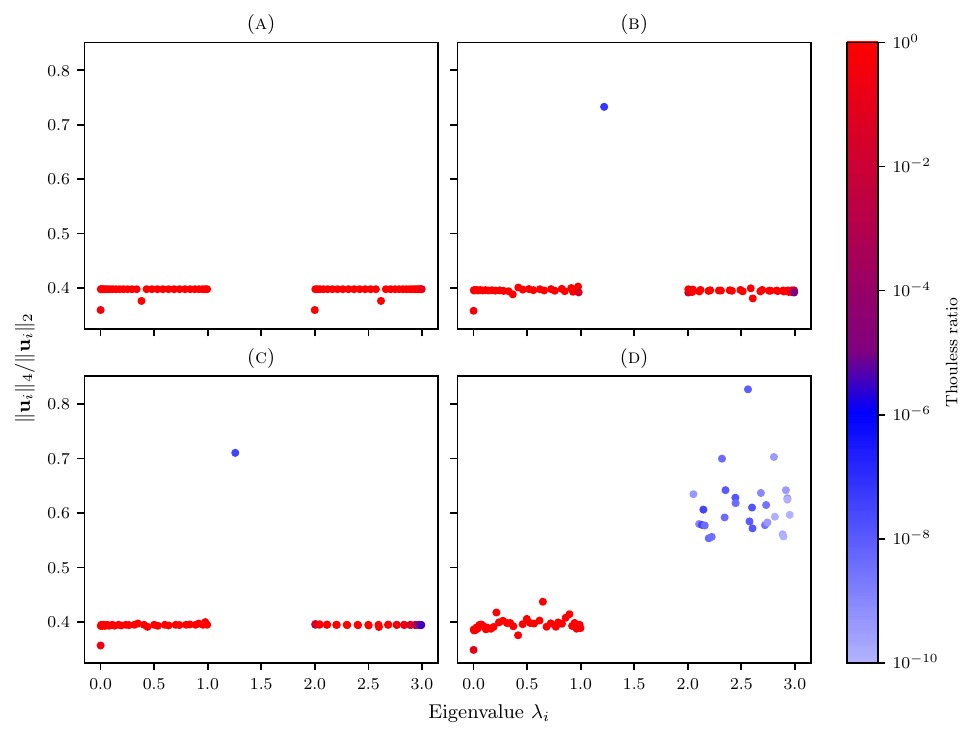}
    \caption{Inverse participation ratio and Thouless ratio for all eigenvalues $\lambda_i$ for the same systems as in \cref{fig:band_function_qpbc}. We can see that the Thouless criterion accurately identifies all localised eigenmodes.} 
    \label{fig:var_vs_loc}
\end{figure}

{
There is an intuitive explanation for the localisation observed in \cref{fig:var_vs_loc}(\textsc{b-d}), that is, \emph{bandgap localisation}. Namely, the modes which are localised in these systems are localised because they lie in the bandgap of one of the constituent blocks. For the two defected systems (\textsc{b}) and (\textsc{c}), the defect blocks in the middle induce a defect mode which lies inside the bandgap $(2,3)$ of the surrounding dimer blocks. This causes the defect mode to be exponentially localised, since the dimer blocks make up all but one of the blocks of the respective systems. 

For system (\textsc{d}) consisting of randomly sampled blocks of either single resonators or dimers the story is similar: Every dimer block introduces two eigenmodes, one is lying in its lower band $(0,1)$ and another in its upper band $(2,3)$, causing the upper band $(2,3)$ to fill up as well. However, every mode in the dimer upper band lies in the bandgap of the single resonator. Due to random sampling, an upper dimer eigenmode must pass through arbitrarily many single resonators, leading to exponential localisation.

In \cref{sec:propagation}, we will investigate these intuitions using a propagation matrix approach.

\begin{figure}[h]
    \centering
    \includegraphics[width=\textwidth]{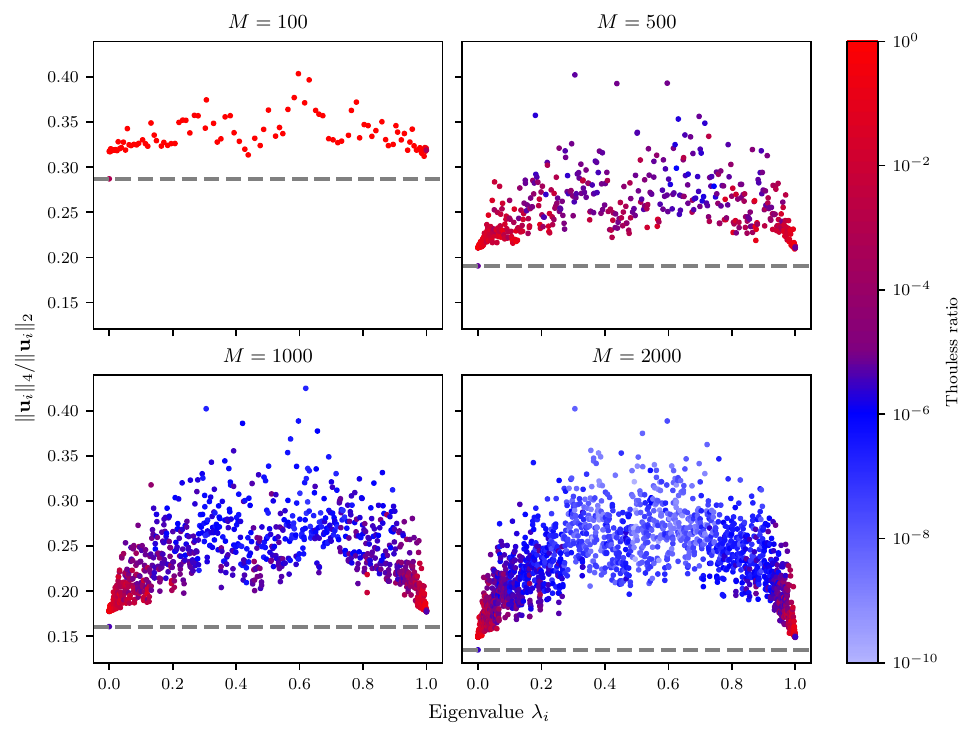}
    \caption{Inverse participation ratio and Thouless ratio for all eigenvalues $\lambda_i$ in the lower band $(0,1)$ for randomly sampled block disordered systems (constructed as in \cref{ex:standard_blocks}) of increasing size. The gray dashed line marks the lowest possible IPR ($1/\sqrt[4]{N}$) for completely delocalised eigenmodes.
    We can see that the modes away from the band edges become increasingly localised as the system size is increased. The Thouless criterion accurately captures this behaviour.} 
    \label{fig:var_vs_loc_increasing size}
\end{figure}
However, there also exists another localisation mechanism, that is, \emph{disorder localisation}. Although it is usually significantly weaker than bandgap localisation, it is inescapable for one-dimensional random disordered systems. Intuitively, the disordered arrangement of the blocks causes any mode (even those supported in both blocks) to be internally reflected an arbitrary amount of times. As a consequence, for very large system sizes $M\approx 1000$, even modes which are supported (\emph{i.e.}, lie in the pass band) of every constituent block become exponentially localised. This is illustrated in \cref{fig:var_vs_loc_increasing size} where we observe that for a random block system consisting of single resonator and dimer blocks as in \cref{ex:standard_blocks} modes in the lower band $(0,1)$ become localised due to disorder. We further observe that the Thouless criterion accurately captures this localisation as well. 

Random block disordered systems offer an excellent setting to study both types of localisation, as we can easily tweak both types of localisation: Bandgap localisation by varying the bandgaps of the constituent blocks and disorder localisation by varying the relative sampling densities of the constituent blocks. In particular, for a $D$-block system, we can maximise the disorder by choosing $p_1=\dots= p_D = 1/D$, where $D$ is the number of constituent blocks.
}

\section{Propagation matrix analysis}\label{sec:propagation}
{
In this section, we describe the spectral properties, most importantly bandgaps, of arbitrary block disordered systems by understanding the propagation properties of their constituent blocks. This type of argument has a long tradition dating back to Saxon and Hutner \cite{saxonhutner} and is especially fruitful in the one-dimensional setting, where we can use the powerful transfer matrix approach. In particular, we will employ \cite{hori.matsuda1964Structure,hori1964Structure} to demonstrate that frequencies that lie in the bandgap of all constituent blocks must also lie in the bandgap of the total system.
}

\subsection{Transfer and propagation matrices}
{
Given an arbitrary block disordered system $\mc D = \bigcup_{j=1}^M B_{\oo{j}} = \bigcup_{i=1}^ND_i$ consisting of blocks $B_1, \dots B_D$ arranged in a sequence $\chi\in \ldm$ and described by a generalised capacitance matrix $\mc C = VC$, our aim is to identify gaps in the spectrum $\sigma(\mc C)$. To that end, we begin by introducing the \emph{transfer matrices}\footnote{Note that for this section, we will assume the wave speeds $v_j=1$ for all resonators.}
\begin{equation}\label{eq:transfermat}
    T_i(\lambda) \coloneqq \left(\begin{array}{cc}
         1+\frac{s_i}{s_{i-1}}-\ell_is_i \lambda& -\frac{s_i}{s_{i-1}} \\
         1& 0
    \end{array}\right),
\end{equation}
for $1\leq i \leq N$ and $\lambda\in \R$. We choose $s_0\coloneqq s_N$, where $s_N$ is given due to the block disordered construction as in \cref{eq:sn}.

The matrix $T_i(\lambda)$ is called the transfer matrix since, for any $\bm u\in \C^N$ satisfying $(\mc C -\lambda I_N)\bm u = 0$, we must have
\begin{equation}\label{eq:transfer}
    \left(\prod_{i=a}^{b-1}T_i(\lambda)\right)
    \begin{pmatrix}
        \bm u^{(a)}\\ \bm u^{(a-1)}
    \end{pmatrix} = \begin{pmatrix}
        \bm u^{(b)}\\ \bm u^{(b-1)}
    \end{pmatrix}, 
\end{equation}
for any $2\leq a<b\leq N$. {Here, we recall that $\bm u^{(c)}$ for $c=a,a-1,b, b-1$ denotes the $c$\textsuperscript{th} entry of $\bm u$.}

Moreover, if we define the \emph{total transfer matrix} as $T_{tot}\coloneqq \prod_{i=1}^{N}T_i$, then we find that 
\begin{equation}\label{eq:eva_transfermat_characterization}
    \lambda\in \sigma(\mc C) \iff T_{tot}(\lambda)\begin{pmatrix}
        1\\1
    \end{pmatrix} = \nu\begin{pmatrix}
        1\\1
    \end{pmatrix} \text{  for some $\nu \in \R$}.
\end{equation}
We have thus reduced the search for the eigenvalues $\lambda$ of $\mc C$ to understanding the dynamics of $T_{tot} = \prod_{i=1}^{N}T_i$.

As a next step, we would like to connect the dynamic properties of $T_{tot}$ to those of the transfer matrices of the constituent blocks. However, this approach is hindered by the fact that $T_i$ is \enquote{nonlocal} in the sense that it contains both information about the $i$\textsuperscript{th} and $(i-1)$\textsuperscript{th} resonator because $T_i$ takes $(\bm u^{(i)},\bm u^{(i-1)})^\top$ as input.  Here, the superscript $\top$ denotes the transpose. 

Ideally, we would like to work with matrices that only move information from the $i$\textsuperscript{th} site to the $(i+1)$\textsuperscript{th} site and thus only require information about the $i$\textsuperscript{th} resonator. 

The solution lies in performing a change of basis from the vectors $(\bm u^{(i)},\bm u^{(i-1)})^\top$ to the evaluation of the continuous solution and its derivative at the resonator $(u(x_i^{\iL}),u'(x_i^{\iL}))^\top$. In particular, for any $\bm u \in \R^N$, we construct $u:\R \to \R$ as in \cref{thm:capapprox} and denote by $u(x_i^{\iL})$ and $u'(x_i^{\iL})$ the exterior limits at the left edge of the $i$\textsuperscript{th} resonator $x_i^{\iL}$ to find that
\begin{equation}\label{eq:cobeq}
    u(x_i^{\iL}) = \lim_{x\uparrow x_i^{\iL}} u(x) = \bm u^{(i)} \quad \text{ and } \quad u'(x_i^{\iL}) = \lim_{x\uparrow x_i^{\iL}} u'(x) = \frac{\bm u^{(i)}-\bm u^{(i-1)}}{s_{i-1}}.
\end{equation}
At this point, we note that this characterisation holds generally in the subwavelength setting and not only when $\bm u$ is a resonant eigenmode. We refer to \cite{feppon.ammari2022Subwavelength} for further details. 
Now, \cref{eq:cobeq} suggests defining the following change of basis matrices: 
\begin{equation}
    \underbrace{\begin{pmatrix}
        1&0\\
        \frac{1}{s_{i-1}}& -\frac{1}{s_{i-1}}
    \end{pmatrix}}_{S_{i-1}\coloneqq } \begin{pmatrix}
    \bm u^{(i)}\\\bm u^{(i-1)}
\end{pmatrix} = \begin{pmatrix}
    u(x_j^{\iL})\\u'(x_j^{\iL})
\end{pmatrix}.
\end{equation}

Performing this change of basis, we find that 
\begin{equation}\label{eq:propmatdef}
    \begin{pmatrix}
        u(x^{\iL}_{i+1})\\
        u'(x^{\iL}_{i+1})
    \end{pmatrix} = S_iT_i(\lambda)S_{i-1}^{-1} \begin{pmatrix}
        u(x^{\iL}_{i})\\
        u'(x^{\iL}_{i})
    \end{pmatrix} = \underbrace{\begin{pmatrix}
        1-s_i\ell_i\lambda & s_i\\
        -\ell_i\lambda & 1
    \end{pmatrix}}_{P_{l_i,s_i}(\lambda)\coloneqq} \begin{pmatrix}
        u(x^{\iL}_{i})\\
        u'(x^{\iL}_{i})
    \end{pmatrix}.
\end{equation}
For general resonator lengths $\ell$, spacings $s$ and frequencies $\lambda$, we will call $P_{\ell,s}(\lambda)$ the \emph{subwavelength propagation matrix}. At this point, we would like to note that the subwavelength propagation matrix $P_{\ell,s}(\lambda)$ can alternatively be obtained by calculating the regular propagation matrix for \eqref{waveeq} and then by taking the leading-order in its asymptotic expansion as $\omega =\BO(\sqrt{\delta})$ and $\delta\to 0$.

From \cref{eq:propmatdef}, we can also see that, in contrast to the transfer matrix $T_i(\lambda)$, the propagation matrix $P_i(\lambda)\coloneqq P_{\ell_i,s_i}(\lambda)$ of the $i$\textsuperscript{th} resonator depends \emph{only} on the properties of that resonator, as desired. 
As a consequence, the \emph{block propagation matrices} $P_{B_d}(\lambda)$ are well defined as
\begin{equation}
    P_{B_d}(\lambda) \coloneqq \prod_{k=1}^{\len(B_d)}P_{\ell_k(B_d), s_k(B_d)}(\lambda),
\end{equation}
and we can decompose the \emph{total propagation matrix} $P_{tot}(\lambda) \coloneqq \prod_{i=1}^NP_i(\lambda)$ as the product of the block propagation matrices
\begin{equation*}
    P_{tot}(\lambda) = \prod_{j=1}^MP_{B_{\oo{j}}}(\lambda).
\end{equation*}

\subsection{Propagation matrices and bandgaps}
Having defined the block propagation matrices $P_{B_d}(\lambda)$ we now aim to relate their properties to the spectrum $\sigma(\mc C)$ of $\mc C$.

We begin by identifying the pass bands and bandgaps of the constituent blocks. To do so, we investigate the eigenvalues of $P_{B_d}(\lambda)$, which we shall denote $\xi_1, \xi_2$ sorted by ascending absolute value $\abs{\xi_1} < \abs{\xi_2}$. Note that although $\xi_1, \xi_2$ are eigenvalues of $P_{B_d}(\lambda)$, they are not related to the resonant frequencies $\lambda_j$. Instead, they control the spatial growth or decay induced by the propagation matrix. 

From \cref{eq:propmatdef}, we can see that $\det P_{\ell,s}(\lambda) = 1$ for any $\ell,s,\lambda\in \R$ and thus also $\det P_{B_d}(\lambda) = 1$ for any block propagation matrix. The eigenvalues of $P_{B_d}(\lambda)$ are thus given by 
\[
    \xi_{1,2} = \frac{\tr P_{B_d}(\lambda) \pm \sqrt{(\tr P_{B_d}(\lambda))^2 - 4}}{2},
\]
leading to the following characterisation:
\begin{equation}
    \begin{cases}
        \xi_1 = \overline{\xi_2} \in \mathbb{S}^1\subset \C & \text{if } \abs{\tr P_{B_d}(\lambda)}\leq2, \\
        0 < \abs{\xi_1} < 1 < \abs{\xi_2} & \text{if } \abs{\tr P_{B_d}(\lambda)}>2.
    \end{cases}
\end{equation}
We can thus see that a periodised arrangement of the block $B_d$ supports oscillatory extended modes if $\abs{\tr P_{B_d}(\lambda)}\leq2$ while any frequency for which $\abs{\tr P_{B_d}(\lambda)}>2$ must be exponentially growing or decaying and can thus not be a resonant frequency of the periodised system. 
It is thus natural to consider a frequency $\lambda\in \R$ in the \emph{pass band} of the block $B_d$ if $\abs{\tr P_{B_d}(\lambda)}\leq2$ and in the \emph{bandgap} if $\abs{\tr P_{B_d}(\lambda)}>2$.

{
We can now use the propagation matrix theory developed in \cite{hori.matsuda1964Structure} to characterise bandgaps of block disordered systems. To avoid edge effects we will study the infinite limit obtained by sampling an infinite sequence $\chi \in \ldz$ and arranging the blocks accordingly in both directions. This infinite system is then characterised by a tridiagonal bounded lattice operator $\mc C_\infty:\ell^2(\Z)\to\ell^2(\Z)$. For such a system we can characterise the spectral bandgaps as follows.
\begin{theorem}\label{thm:saxonhutner}
    Consider an infinite block disordered system with blocks $B_1, \dots, B_D$, sequence $\chi \in \ldz$ and a corresponding generalised capacitance operator $\mc C_\infty$. If a given frequency $\lambda\in \R$ lies in the bandgap of all constituent blocks\footnote{Furthermore, an additional assumption on the propagation matrix eigenvectors ensuring the existence of invariant cones is required. We refer to \cite{hori.matsuda1964Structure} for the details.}, that is,
    
    \[
        \abs{\tr P_{B_d}(\lambda)} > 2 \quad \text{ for all }d=1,\dots, D,
    \]
    then $\lambda$ must also lie in the bandgap of the entire system, \emph{i.e.} $\lambda\notin \sigma(\mc C_\infty)$.
\end{theorem}
\begin{proof}
    We first note that \cref{eq:transfer} continues to hold in the infinite case and $T_i(\lambda)$ for $ i\in \Z$ are thus the transfer matrices of the infinite system. Rearranging \cref{eq:propmatdef} into $T_j(\lambda) = S_i^{-1}P_i(\lambda)S_{i-1}$ yields
    \[
        \prod_{i=a}^b T_i(\lambda) = S_b^{-1}\prod_{i=a}^bP_j(\lambda)S_{a-1}
    \]
    ensuring $T_i$ and $P_i$ exhibit similar dynamics. 

    We can then collect the terms $P_j$ by blocks to find $P_{B_{\oo{j}}}(\lambda)$ for $j\in \Z$. The result then follows by applying \cite[Theorem I]{hori.matsuda1964Structure} to the sequence $P_{B_{\oo{j}}}(\lambda)$.
\end{proof}}

\begin{figure}
    \centering
    \includegraphics[width=0.95\textwidth]{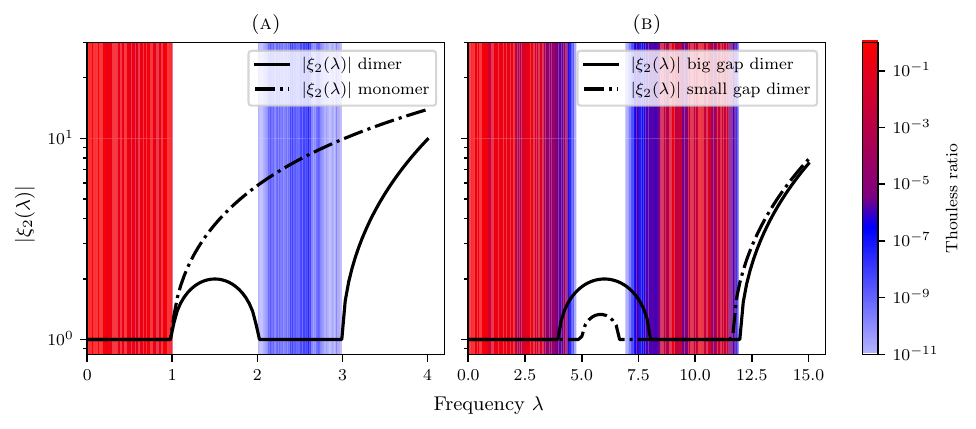}
    \caption{Comparison of maximal block propagation matrix eigenvalue $\abs{\xi_2(\lambda)}$ and Thouless ratios $g(\lambda_i)$ of the entire random block disordered system consisting of these block (each coloured vertical line corresponds to the Thouless ratio $g(\lambda_i)$ of the eigenvalue $\lambda_i$). Where there are no such lines, the density of eigenvalues is zero. 
    For each subplot, the random block disordered system was constructed by sampling $M=100$ of the respective blocks. \textbf{Left:} Single resonator (monomer) blocks and dimer resonator blocks as in \cref{ex:standard_blocks}. \textbf{Right:} Two distinct dimer blocks.}
    \label{fig: var and eigs prop matrix}
\end{figure}

This theorem is illustrated in \cref{fig: var and eigs prop matrix}. We plot the larger eigenvalue $\xi_2(\lambda)$ as a function of the frequency $\lambda\in \R$ for a variety of blocks. Furthermore, we construct a random block disordered system using the respective building blocks and $M=100$ samples, and for each eigenfrequency $\lambda_i\in \sigma(\mc C)$ we also plot the corresponding Thouless ratio $g(\lambda_i)$ as a vertical line. { Although only formulated for the infinite case $\chi\in \ldz$, empirically this theorem also gives an accurate prediction of the spectrum in the finite case $\chi\in \ldm$.} In subplot (\textsc{a}), we consider the standard single resonator (monomer) and dimer blocks introduced in \cref{ex:standard_blocks}, while in subplot (\textsc{b}), we consider two different dimer blocks with distinct band structures. We can see that \cref{thm:saxonhutner} accurately predicts the bandgaps of the respective random block disordered systems. 

Moreover, we can see that in the block random setting, the prediction is even stronger. That is, a random block disordered system exhibits a bandgap at $\lambda$ if \emph{and only if} $\lambda$ is in the gap of all constituent blocks, going beyond \cref{thm:saxonhutner}. The \emph{only if} part stems from the fact that, since the blocks are sampled independently, any arbitrarily long sequence of the same block repetitions will occur with probability $1$ as $M\to \infty$. These same-block sequences contribute with eigenvalues wherever their corresponding block has a band, and thus for any frequency $\lambda$ lying in the pass band of any of the constituent blocks, we will find resonant frequencies arbitrarily close to $\lambda$ as $M\to \infty$. 

When the eigenfrequency $\lambda_i$ lies in the pass band of one of the blocks and the bandgap of the others, we can see that the Thouless ratio $g(\lambda_i)$ is much lower compared to the Thouless ratios of the eigenfrequencies lying in the shared pass band. This indicates that the eigenmodes found there exhibit a significantly stronger localisation. Nevertheless, the density of states in these regions is nonzero, distinguishing them from the bandgaps of the total system. Intuitively, each block for which this region does not lie in the bandgap contributes eigenvalues to this region -- but because these eigenvalues lie in the gap of some other constituent blocks, they become strongly localised. 

Since such regions display a mixture of behaviours observed in both the shared pass band and the bandgap (and for reasons which will be examined in the following subsection), we shall call such regions \emph{hybridisation regions}. We obtain the following complete description of the spectral regions of block disordered systems: 
A frequency $\lambda$ lies in the
\begin{description}
    \item[Shared pass band] if and only if it is in the band of all constituent blocks;
    \item[Bandgap] if and only if it is in the bandgap of all constituent blocks;
    \item[Hybridisation region] otherwise.
\end{description}
In the bandgap the density of eigenmodes is zero and only defect modes can occur, while both the shared pass band and hybridisation regions support resonant modes -- though at strongly distinct degrees of localisation as seen in \cref{fig: var and eigs prop matrix}.

\subsection{Lyapunov exponent prediction}
Finally, in this section, our aim is to relate the \emph{Lyapunov exponent} of the total system to those of the constituent parts, as this will allow us to predict the decay of a given mode at a given frequency $\lambda$. 
The \emph{Lyapunov exponent} is defined as the norm of the total transfer matrix $T_{tot}(\lambda)$ over the total number of resonators $N$: 
\begin{equation}\label{eq:lyapunov}
    \gamma(\lambda) = \frac{1}{N}\ln \norm{T_{tot}(\lambda)},
\end{equation}
and can thus be seen as describing \enquote{typical} decay per resonator.

We now aim to predict the Lyapunov exponent for large random block disordered systems with block probabilities $p_1, \dots, p_D$ at frequencies $\lambda$ such that the assumptions of \cref{thm:saxonhutner} are satisfied, \emph{i.e.} $\lambda$ lies in the bandgap of all constituent blocks. In the following, we will use $\rho(A)$ to denote the \emph{spectral radius} of the matrix $A$ and obtain the approximation
\begin{equation}\label{eq:lyapunov_estimate}
    \begin{gathered}
    \gamma(\lambda) 
    = \frac{1}{N}\ln \norm{T_{tot}(\lambda)} 
    \approx \frac{1}{N}\ln \rho(T_{tot}(\lambda)) 
    = \frac{1}{N}\ln \rho(P_{tot}(\lambda)) \\
    \approx \frac{1}{N} \sum_{j=1}^M \ln \rho(P_{B_{\oo{j}}}(\lambda)) = \frac{1}{M}\sum_{j=1}^M\frac{M}{N} \ln \rho(P_{B_{\oo{j}}}(\lambda)) \\
    \approx \E \frac{M}{N} \ln \rho(P_{B_{\oo{0}}}(\lambda))
    =\frac{\sum_{d=1}^Dp_d \ln \rho(P_{B_d}(\lambda))}{\sum_{d=1}^Dp_d \len(B_d)}.
\end{gathered}
\end{equation}
\begin{figure}
    \centering
    \includegraphics[width=0.95\textwidth]{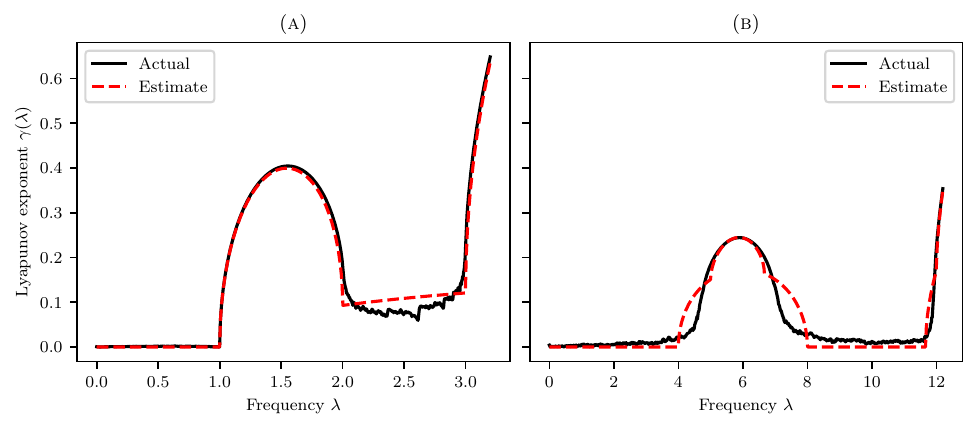}
    \caption{Actual Lyapunov exponent $\gamma(\lambda)$ versus the estimate obtained from \cref{eq:lyapunov_estimate} for large ($M=10^5$) random block disordered systems. \textbf{Left:} Single resonator (monomer) blocks and dimer resonator blocks as in \cref{ex:standard_blocks}. \textbf{Right:} Two distinct dimer blocks.}
    \label{fig:lyapunov_vs_estimate}
\end{figure}

Here, the first approximation originates from the fact that $$\frac{1}{N}\ln \norm{T_{tot}(\lambda)} = \frac{1}{N}\ln \kappa(T_{tot}(\lambda))\rho(T_{tot}(\lambda)),$$ where $\kappa(T_{tot}(\lambda))$ is the condition number of $T_{tot}(\lambda)$. The assumptions of \cref{thm:saxonhutner} ensure that $\kappa(T_{tot}(\lambda)) = \mc O(1)$ as $N\to \infty$, which shows that the approximation error is of the order $\mc O(N^{-1})$. 

The second approximation hinges on the fact that as we get deeper into the bandgap, the block propagation matrices become increasingly hyperbolic (\emph{i.e.} $\abs{\tr P_{B_d}(\lambda) \gg 2}$) and their eigenspaces tend to line up allowing for the splitting { $\rho(P_{tot}(\lambda)) \approx \prod_{j=1}^M\rho(P_{\oo{j}}(\lambda))$}. 

Finally, using the fact that the blocks are sampled randomly and independently with probabilities $p_1,\dots, p_D$ and applying the law of large numbers,  we obtain the last approximation.

In \cref{fig:lyapunov_vs_estimate}, we compare the estimated Lyapunov exponent $\overline{\gamma}(\lambda)$ obtained using \cref{eq:lyapunov_estimate} with the actual Lyapunov component for large random block disordered systems. The random systems considered are the same as in \cref{fig: var and eigs prop matrix} and we see that the estimate accurately matches the actual Lyapunov exponent where the frequency $\lambda$ lies in the shared bandgap $(1,2)$ for system (\textsc{a}) and approximately $(5,6.6)$ for system (\textsc{b})). In the next section, we employ this estimate to accurately predict the decay of a defect mode in the bandgap.

It is also worth examining why this prediction fails in the other regions, as this is connected to a variety of fundamental phenomena of disordered systems:

In the shared pass band regions where $\abs{\tr P_{B_d}(\lambda)}<2$ for all $d=1,\dots ,D,$ we necessarily estimate $\overline{\gamma}(\lambda)=0$ . However, recall that in \cref{fig:var_vs_loc_increasing size} we observed localisation even for eigenfrequencies $\lambda_i\in (0,1)$ in the shared pass band of system (\textsc{a}) due to \emph{disorder localisation} which induces (comparatively weak) exponential decay for all eigenmodes of a one-dimensional system. 
Because estimate \cref{eq:lyapunov_estimate} only considers \emph{bandgap decay}, \emph{i.e.} whether $\lambda$ lies in the gap of one of the constituent blocks, this leads to a systematic underestimation of the actual small but nonzero Lyapunov exponent in such shared pass band regions. 

In contrast, in hybridisation regions (\emph{ex.} $(2,3)$ for (\textsc{a}) or approximately $(6.6,8)$ for (\textsc{b})), estimate \cref{eq:lyapunov_estimate} systematically \emph{overestimates} the actual Lyapunov exponent. This has an interesting explanation due to the random \emph{i.i.d.} sampling of the resonator blocks: Due to this sampling, every local arrangement of blocks occurs arbitrarily often as the number $M$ of sampled blocks increases. Consequently, every time a strongly localised eigenmode $\bm u_j$ in the hybridisation region encounters a sequence of blocks that closely resembles the sequence at its original localisation, hybridisation occurs, causing a local increase in the magnitude of the eigenmode. This hybridisation prevents $\bm u_j$ from realising the full decay predicted by estimate \cref{eq:lyapunov_estimate} and motivates naming such eigenmodes \emph{hybridised bound states} and the regions where they occur \emph{hybridisation regions}. We note that while this hybridisation prevents the modes from achieving maximal decay, it is still too weak to induce delocalisation, as the expected distance before a matching sequence is encountered increases exponentially with the desired accuracy of the match.

\section{Defect frequencies}  \label{sec:defects}
Having identified a method of creating strongly disordered systems with well-defined bands and bandgaps, we are now in a position to investigate localised defect eigenmodes in the disordered setting.
There are multiple strategies on how to modify a periodic structure of subwavelength resonators with a bandgap to support localised defect eigenmodes \cite{ammari.davies.ea2021Functional,cbms}. In this section, we consider one of these --- the most suited for disordered systems --- and show that it generates localised eigenmodes in disordered structures.

We consider a disordered structure of any type known to have a bandgap predicted by \cref{thm:saxonhutner} and seek to slightly modify the system to obtain (at least) one localised eigenmode centred on some specific resonator, indexed by $i_{d}$.

The creation of a localised eigenmode occurs by modifying the wave speed inside the resonator $i_d$. This approach is particularly suited to a disordered system because it does not rely on any geometrical assumption. Mathematically, this translates to considering the eigenvalues of the generalised capacitance matrix $\mathcal{C}=VC$ with
$V$ instead of being defined as in (\ref{def:v}) is now given by
\begin{align*}
    V := \begin{pmatrix} \ds \frac{v_1^2}{\ell_1} 
& &  \\
	  & \ddots & \\
      &&\frac{(v_{i_d} + \eta)^2}{\ell_{i_d}} \\
      && &\ddots\\
	   &  &&& \ds \frac{v_N^2}{\ell_N}
	\end{pmatrix}.
\end{align*}

In principle, there is no reason why this procedure should generate localised eigenmodes also for disordered systems; after all, there is no clear (periodic) structure that is perturbed. However, the underlying bandgap structure supports the presence of these eigenmodes. 

In \cref{fig: defect mode prediction}, we show the result of this procedure on a structure composed of disordered blocks as in \cref{fig: var and eigs prop matrix}(\textsc{a}). Two eigenvalues jump into the bandgap and their corresponding eigenmodes become localised around the index $i_d$ at which the perturbation is performed. We note that the two eigenmodes have very different decay rates. As in periodic systems, this is due to the different distances that the two eigenvalues in the gap have from the bands.

\begin{figure}
    \centering
    \includegraphics[width=\textwidth]{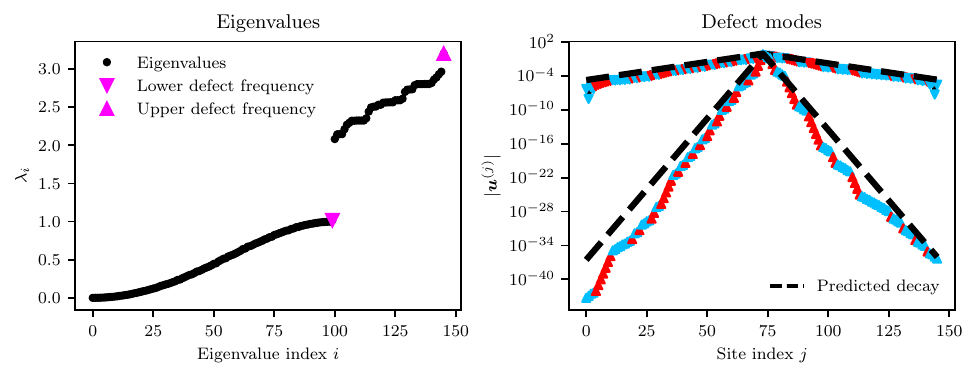}
    \caption{Defect modes for a disordered system as in \cref{ex:standard_blocks} with $M=100$ blocks and a defected dimer in the middle ($i_d=N/2, \eta=0.2$).
    \textbf{Left:} Spectrum of the defected system with two defect frequencies in the bandgap. \textbf{Right:} The two eigenmodes corresponding to the defect frequencies. The upper / lower triangle marks correspond to the upper / lower defect frequency and the mark colour indicates whether the respective site belongs to a single resonator block (red) or dimer block (blue). The eigenmodes are exponentially localised about the defected site, with their rate of decay closely matching the predicted decay (dashed line) obtained by estimate \cref{eq:lyapunov_estimate}.
    The outline of the localised eigenvectors is even more accurately predicted, up to an error of $10^{-13}$, by the discrete Green function defined in \eqref{eq: discrete green function}.}
    \label{fig: defect mode prediction}
\end{figure}

\subsection{Defect eigenmode prediction}
The capacitance matrix approximation of the system delivers the following modal decomposition of the Green function $G(\lambda)$ associated with the system of resonators:
\begin{align}
    G(\lambda) = (\mathcal{C}-\lambda \, I_N)^{-1} = \sum_{i=1}^N\frac{\bm{u}_i \overline{\bm{u}_i}^\top}{\lambda - \lambda_i}, 
    \label{eq: discrete green function}
\end{align}
for $\lambda$ not being an eigenvalue of $\mathcal{C}$, 
where $I_N$ is the $N\times N$ identity matrix (with $N$ being the number of resonators) and $(\lambda_i, \bm{u}_i)$ are the eigenpairs of $\mathcal{C}$.

Away from the spectrum of $\mathcal{C}$, we can thus predict the outline of the eigenmode associated with an excitation at the $i$\textsuperscript{th} resonator by $G(\lambda) \bm{e}_i$, where $(\bm{e}_i)_{i=1}^N$ denotes the standard basis of $\R^N$. The absolute error in predicting the defect eigenmode $v_d$ via $G(\lambda_d) \bm{e}_{i_d}$ is of the order of $10^{-13}$ with $\lambda_d$ being the defect eigenfrequency. Remark that for this computation, one can use the unperturbed generalised capacitance matrix, providing an accurate \emph{a priori} estimation of the defect eigenmodes. However, the defect eigenfrequency $\lambda_d$ still needs to be computed.

The modal decomposition \eqref{eq: discrete green function} can generally be used to show that, as in the periodic case, the defect midgap eigenvalues are associated with localised eigenmodes. This follows from the fact that the inverse of the tridiagonal matrix $\mathcal{C}-\lambda \, I_N$ when $\lambda$ is in a bandgap of $\sigma(\mc C)$ has exponentially decaying entries \cite{jaffard}.

At the beginning of this section, we have commented on the different decay rates of the defect eigenvectors shown in \cref{fig: defect mode prediction}. This effect turns out to be of an interesting origin. In \cref{fig:defect_mode_bandgap_crossing}, we show that the position of the defect eigenvalue depends on which of the blocks is chosen for the perturbation. Recall that we are considering a disordered block structure composed of two types of blocks --- dimer blocks and monomer blocks --- arranged randomly. In the left part of \cref{fig:defect_mode_bandgap_crossing}, we see that if we perturb one of the single resonator blocks, the upper defect eigenmode stays fixed while the lower defect eigenmode moves across the bandgap. For larger perturbations, the lower eigenmode ceases to move, and the upper defect eigenmode moves up instead. If we instead perturb a dimer block, we see the opposite picture. Perturbation causes the upper defect eigenmode to move into the upper gap while the lower defect eigenmode barely moves. We remark furthermore that if one happens to know beforehand that the outermost eigenvalue of a band has an eigenmode localised around some resonator, then perturbing that precise resonator makes the eigenvalue immediately jump into the bandgap avoiding the initially flat part of the blue curve in \cref{fig:defect_mode_bandgap_crossing}(\textsc{b}). Note that in \cref{fig: defect mode prediction} we perturb a dimer block, inducing a large jump for the upper defect eigenvalue and a small jump for the lower one. This is reflected in the two very different decay rates of the corresponding eigenmodes. 

Once the defect frequency is known, we can employ the Lyapunov exponent estimate \cref{eq:lyapunov_estimate} obtained in the previous section to estimate the decay rate of the corresponding defect mode.
As we can see in \cref{fig: defect mode prediction}, this prediction matches the actual decay rate very closely and thus also explains the distinct defect mode decay rates: The upper defect frequency is significantly deeper into the bandgap which corresponds to a larger Lyapunov exponent, as predicted by \cref{eq:lyapunov_estimate} and thus stronger decay.

\begin{figure}[h]
    \centering
    \includegraphics[width=0.9\linewidth]{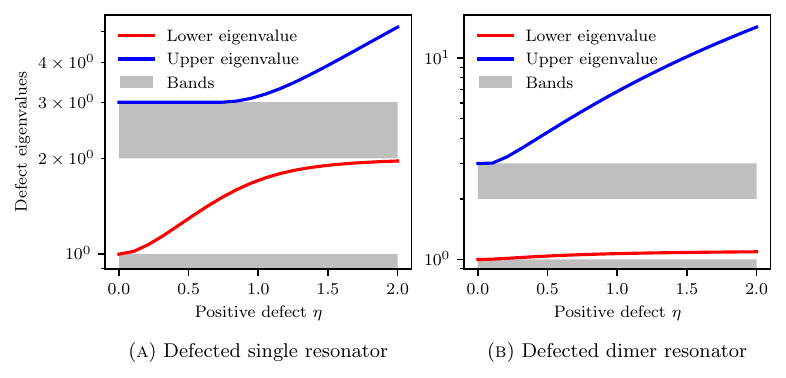}
    \caption{Defect eigenfrequency as a function of positive material parameter perturbation strength. We are perturbing the structure from \cref{fig: var and eigs prop matrix}(\textsc{a}). \textbf{Left:} Perturbation of a single resonator block creating a large jump in the lower eigenvalue. \textbf{Right:} Perturbation of a dimer resonator block creating a large jump in the upper eigenvalue.}
    \label{fig:defect_mode_bandgap_crossing}
\end{figure}

\subsection{Dimer Defect}
Another interesting effect recently observed in \cite{ammari.davies.ea2024Anderson} for the periodic case is the \emph{repulsion effect} occurring when two adjacent resonators are perturbed by the procedure illustrated at the beginning of the section. Level repulsion means that when randomness is added to the entries of a matrix, the eigenvalues tend to separate. In periodic systems, 
such a perturbation induces two coupled eigenmodes to jump into the bandgap and to be localised once with a dipole (odd) symmetry and once with a monopole (even) symmetry. We analyse this situation for disordered systems. In \cref{fig: double defect} we plot the two largest eigenvalues when such a perturbation is performed, observing the same level of repulsion seen in periodic structures.
\begin{figure}
    \centering
    \begin{subfigure}[t]{0.32\textwidth}
        \centering
        \includegraphics[height=0.8\textwidth]{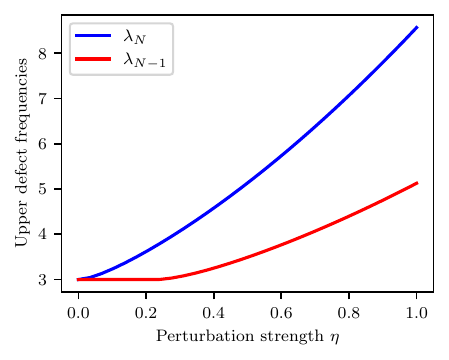}
    \caption{Perturbation of site indices $55$ and $56$.}\label{subfig: doubledeft next}\end{subfigure}\hfill
    \begin{subfigure}[t]{0.32\textwidth}
        \centering
        \includegraphics[height=0.8\textwidth]{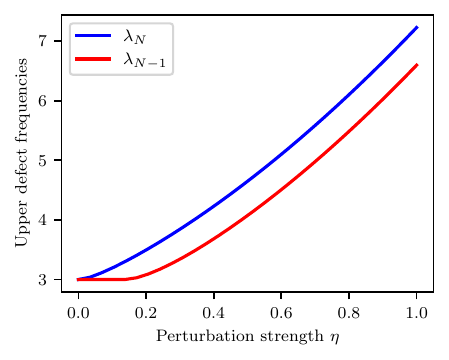}
    \caption{Perturbation of site indices $55$ and $57$.}\label{subfig: doubledeft after next}\end{subfigure}
    \hfill
    \begin{subfigure}[t]{0.32\textwidth}
        \centering
        \includegraphics[height=0.8\textwidth]{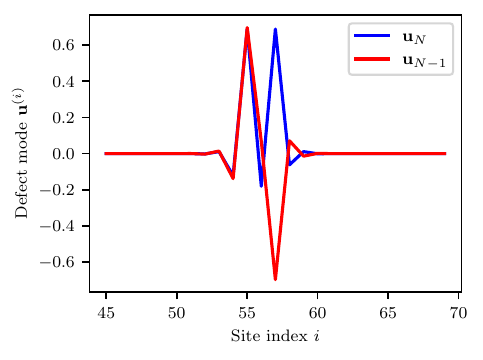}
    \caption{Eigenvectors showing once a monopole and once a dipole nature.}\label{subfig: doubledeft eve}\end{subfigure}
    \caption{When perturbing juxtaposed resonators, two eigenvalues jump into the bandgaps. The distance between these eigenvalues depends on the distance between the perturbed resonators. In any case, the eigenvectors display opposed nature.}
    \label{fig: double defect}
\end{figure}

\section{Stability analysis}  \label{sec4:AL}
A natural question that arises when studying the bands and
bandgaps of disordered systems is whether they are stable under perturbation.
To that end, in line with the previous section, we investigate the behaviour of our disordered systems as the wave speeds inside the resonators are perturbed globally.


In the unperturbed case, all resonators have the same wave speed $v_i=1$. We now globally perturb the wave speed as follows:
\begin{equation*}
    v_i' = v_i + \eta_i,
\end{equation*}
where the perturbations $\eta_i\sim \mathcal{U}[-\sigma,\sigma]$ are independently drawn from the uniform distribution $\mathcal{U}[-\sigma,\sigma]$ with support in $[-\sigma,\sigma]$. In particular, we note that such perturbations destroy the prior block structure of our disordered block systems.

This setup resembles the one studied in \cite{ammari.davies.ea2024Anderson} for the purely periodic case. Indeed, many observations from the periodic case still hold in the disordered bandgapped setting. \cref{fig:global_perturbation} demonstrates that the two statistical phenomena that dominate the periodic case continue to do so in the disordered case. Namely, level repulsion pushing the eigenvalues apart, as demonstrated by the increasing mean eigenvalue separation and Anderson localisation causing a complete localisation of the eigenmodes, as demonstrated by the decreasing mean band variation and increasing degree of localisation.

\begin{figure}[h]
    \centering
    \begin{subfigure}[t]{0.4\textwidth}
        \centering
        \includegraphics[width=\textwidth]{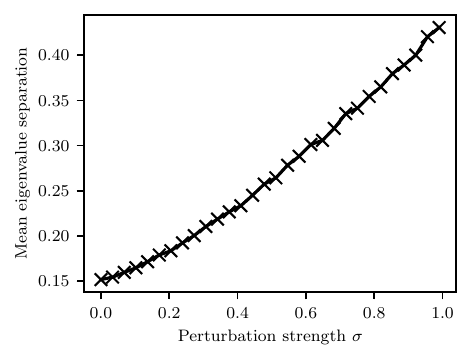}
    \caption{Level Repulsion: Global mean eigenvalue separation as a function of the perturbation strength $\sigma$.}\end{subfigure}\hfill
    \begin{subfigure}[t]{0.5\textwidth}
        \centering
        \includegraphics[width=\textwidth]{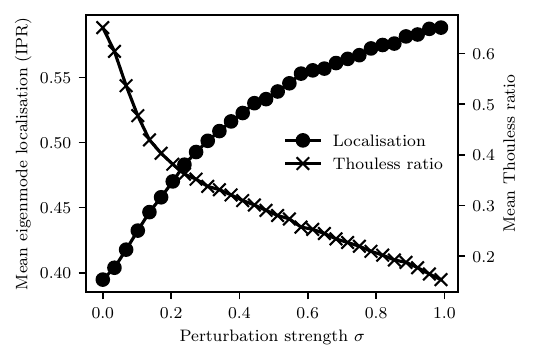}
    \caption{Anderson localisation: Mean eigenmode localisation $\norm{\bm{u}}_4/\norm{\bm{u}}_2$ and mean band variation as a function of the perturbation strength $\sigma$.}\end{subfigure}
    \caption{Level repulsion and Anderson localisation under material parameter perturbation (averaged over $100$ random realisations). The material parameters of the double pass band system from \cref{fig: var and eigs prop matrix}(\textsc{b}) are perturbed uniformly with perturbation strength $\sigma$.}
    \label{fig:global_perturbation}
\end{figure}

These two forces are illuminated further in \cref{fig: dos_variance_histogram}. The starting point is the disordered system with two pass bands introduced in \cref{fig: var and eigs prop matrix}(\textsc{a}). Without perturbation, we can see that the density of eigenfrequencies is evenly distributed into the two bands, with slightly higher density at the band edges. For small perturbations, the bandgap remains unaltered; this is due to Weyl's spectral stability result and the inherited Hermitian nature of the problem. As the perturbation strength increases, the band edges begin to dissolve, as level repulsion causes the eigenvalues to spread out, causing a progressive closing of the bandgap. This also illustrates why the lower band is more resistant to perturbation, as the eigenvalues cannot spread out below zero. At the same time, Anderson localisation causes the eigenmodes to localise, where again the eigenmodes in the upper band are more susceptible. Furthermore, for large perturbation strength, the bandgap vanishes and the majority of eigenmodes become fully localised as the system transitions from a disordered but structured system to a fully random and unstructured one.
\begin{figure}[h]
    \centering
    \includegraphics[width=\textwidth]{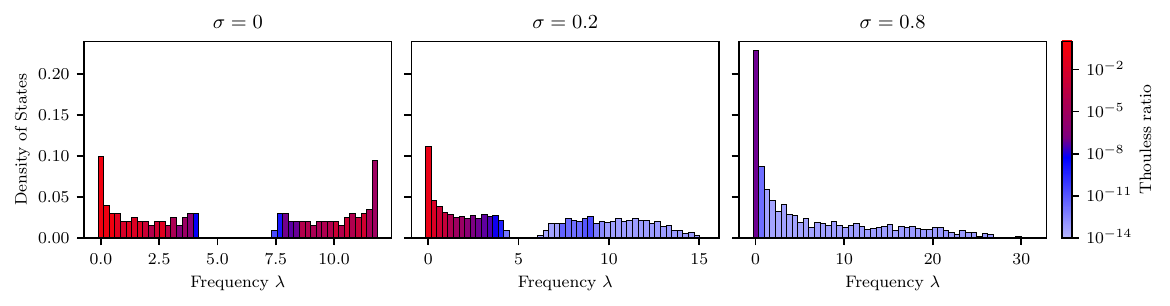}
    \caption{Density of eigenfrequencies and band variation under material parameter perturbation (averaged over $100$ random realisations). We again perturb the double pass band system from \cref{fig: var and eigs prop matrix}(\textsc{b}) uniformly. The colours mark the mean band variation in the corresponding bin. As the perturbation strength increases, the density of eigenfrequencies spreads out, converging to the expected $e^{-x}$ distribution \cite{ammari.davies.ea2024Anderson} and the eigenmodes become localised.}
    \label{fig: dos_variance_histogram}
\end{figure}

Finally, in \cref{fig: global_perturbation_heatmap} we compare the perturbation response to the periodic dimer system. For all three systems, the observed behaviour mirrors the one observed in \cref{fig: dos_variance_histogram} as the eigenvalues spread out and the eigenmodes become localised. In particular, the global perturbation does not appear to make a distinction between eigenmodes in the shared pass bands or not, both becoming completely {localised} for large perturbation strengths.

\begin{figure}[h]
    \centering
    \includegraphics[width=\textwidth]{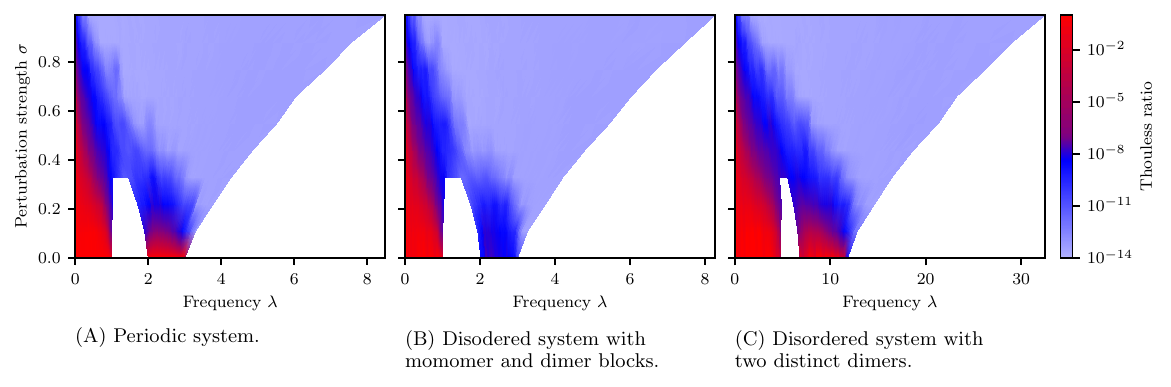}
    \caption{Band variation and spectral distribution of different systems under material parameter perturbation (averaged over $100$ random realisations). We perturb the periodic dimer system from \cref{fig:band_function_qpbc}(\textsc{a}) and the single and double pass band disordered systems from \cref{fig: var and eigs prop matrix}(\textsc{a} and \textsc{b}) uniformly. The colours mark the band variation of the eigenmode at the corresponding point. The behaviour is the same in all cases, as the eigenvalues spread out and the corresponding eigenmodes become increasingly more localised.}
    \label{fig: global_perturbation_heatmap}
\end{figure}

The material parameter perturbation is general in the sense that both spacing and resonator size perturbations cause the same localisation behaviour. This is not surprising, given that perturbations in the material, spacing, and size of the resonators induce equivalent perturbations in the governing generalised capacitance matrix.

\section{Prediction of wave localisation positions and mobility edges}
\label{sec:prediction}
In this section, we introduce two other tools from the study of disordered systems, namely the landscape function and fractal dimension, to gain additional perspectives on wave localisation and mobility edges in the block disordered setting.

\subsection{Landscape function for predicting wave localisation positions}
We first consider the landscape function. { Our aim here is to compare its performance for the detection of localised eigenmodes in perturbed periodic systems with that in disordered systems when hybridised bound eigenmodes occur. As will be numerically shown, the landscape function detects only the localised eigenmodes corresponding to eigenfrequencies that are in the bandgap region. 
It is not sensitive enough to the localised eigenmodes that are in the hybridisation region.}  

The landscape function first introduced in \cite{filoche} has already been used to predict wave localisation in subwavelength resonator systems \cite{davies.lou2023Landscape}. The prediction is through the computation of the  solution $\bm u$ of
\begin{align}
    VC \bm u = \bm 1,
    \label{eq: localisation landscape}
\end{align}
where $\bm 1$ is the vector of only ones. We note that in one dimension this procedure is complicated by the fact that $VC$ has a non-trivial kernel --- to resolve this, we instead compute $\bm u$ to be the least-squares solution of \cref{eq: localisation landscape}. 

In \cref{fig:landscape}, where the entries of the solution $\bm u$ of \eqref{eq: localisation landscape} for defected periodic or block disordered systems are plotted, we see that the localisation landscape can be used for the detection of localised eigenmodes even in the presence of hybridised bound states which are strongly localized as well. 
\begin{figure}[h]
    \centering
    \begin{subfigure}[t]{0.47\textwidth}
        \centering
        \includegraphics[width=\textwidth]{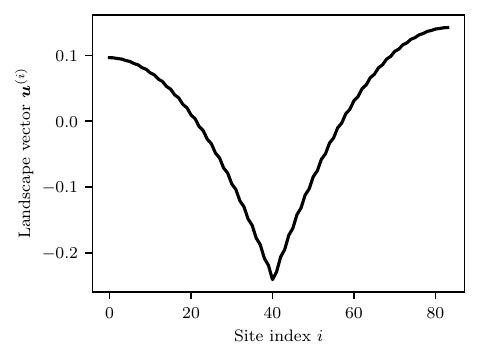}
    \caption{Periodic structure of dimers with material parameter defect at the site index 40.}\end{subfigure}\hfill
    \begin{subfigure}[t]{0.47\textwidth}
        \centering
        \includegraphics[width=\textwidth]{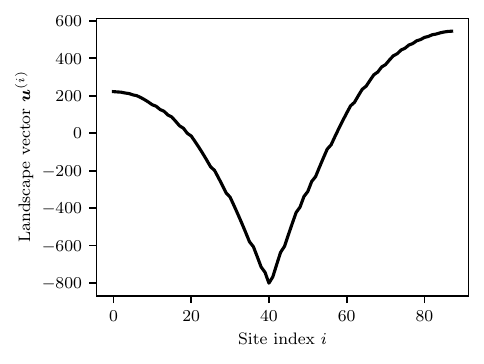}
    \caption{Random block disordered system as in \cref{fig: var and eigs prop matrix}(\textsc{a}) with additional material parameter defect at the site index 40.}\end{subfigure}\hfill
    \caption{The landscape function predicts the localisation induced by material parameter defects in disordered systems, even in the case when strongly localised hybridised bound states are present.}
    \label{fig:landscape}
\end{figure}
By decreasing the material parameter in the defected resonator, a localised eigenmode emerges from the edge of the second pass band into the bandgap of the system. This causes the entries of the landscape function $\bm u$ to achieve their minimum at the location of the defected resonator. 
It should be noted that the landscape function $\bm u$ is unaffected by the localisation of the hybridised bound states in the hybridisation region of the system considered in \cref{fig:landscape}(\textsc{b}).

As discussed in \cite{davies.lou2023Landscape,crossover}, one of the limitations of landscape theory is that it only reveals useful information about the lowest eigenmodes (or, after translation, the highest eigenmodes). Here, the shape of the localised eigenmode appearing in the middle of the spectrum is very nicely predicted by the landscape function. In some sense, the landscape function is better suited for detecting localised eigenmodes in systems with multiple pass bands than in monomer systems.
Such behaviour of the landscape function may be rigorously justified based on Chebyshev polynomials, in a similar way to \cite{ammari.barandun.ea2024Exponentially}.

\subsection{Fractal dimension of localised eigenmodes for predicting mobility edges}

Finally, we show that the behaviour of the fractal dimension related to the localised eigenmodes of randomly perturbed systems can be used to identify the mobility edges in disordered systems, even in the presence of hybridisation regions.

Specifically, in this framework, the fractal dimension of an eigenvector $\bm{u} \in\R^N$ of the generalised capacitance matrix is defined as
    \begin{align}
    L(\alpha,\bm{u}) \coloneqq - \ln (\vert\{i:(\bm{u}^{(i)})^2<N^{-\alpha}\}\vert)
    \label{eq: fractal dim}
\end{align}
for $\alpha>0$, {where $\vert \; \vert$ denotes the cardinality of the set}; see, for instance, \cite{fractaldimension}. We note two interesting values for $L$: $L(\alpha,\bm 1 / \sqrt{N})=-\ln (N)$ for $\alpha<1$ and $\infty$ for $\alpha\geq1$, and $L(\alpha,\bm{e}_i)=-\ln (N-1)$ for every $\alpha$ where $\bm{e}_i$ is an element of the standard basis.

\begin{figure}[h]
    \centering
    \begin{subfigure}[t]{0.45\textwidth}
        \centering
        \includegraphics[width=\textwidth]{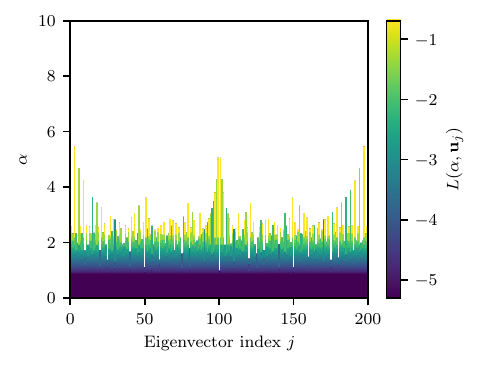}
    \caption{Periodic system.}\end{subfigure}\hspace{5mm}
    \begin{subfigure}[t]{0.45\textwidth}
        \centering
        \includegraphics[width=\textwidth]{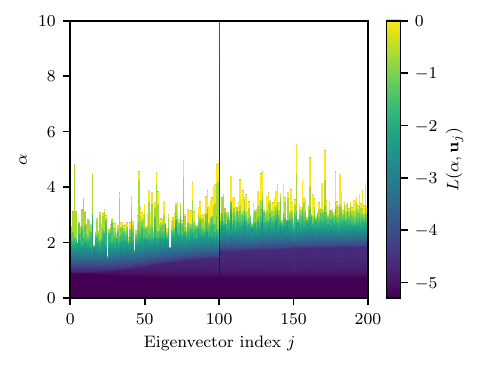}
    \caption{Periodic system with single material parameter defect at index 100.}\end{subfigure}\\
    
    \begin{subfigure}[t]{0.45\textwidth}
        \centering
        \includegraphics[width=\textwidth]{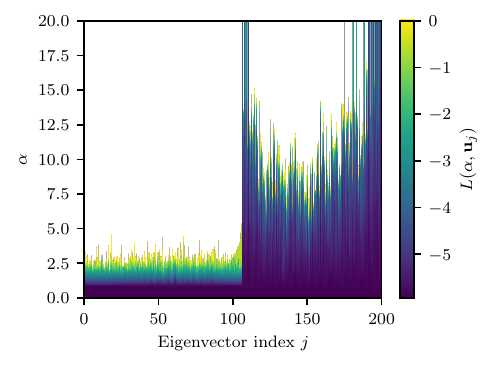}
    \caption{Disordered system from \cref{fig: var and eigs prop matrix}(\textsc{a}).}\end{subfigure}
    \caption{The fractal dimension detects mobility edges in periodic and disordered structures. A clear phase change between the eigenmodes in the shared pass band and those in the hybridisation region can be observed in (\textsc{c}). The eigenvectors are arranged in ascending order with respect to their corresponding eigenvalues.}
    \label{fig: fractal dimension}
\end{figure}
As shown in \cref{fig: fractal dimension}, the fractal dimension detects the mobility edges in both the periodic (\textsc{a}) and disordered (\textsc{c}) case as well as the defect mode in the periodic system with defect (\textsc{b}). 
When comparing the periodic and disordered systems, we note that the fractal dimension of the eigenmodes in the shared pass band of the disordered system closely resembles that the fractal dimension of the delocalised modes in the periodic system. In contrast, the hybridised bound states in the hybridisation region of the disordered system display significantly higher fractal dimension -- owing to their strong localisation. Nevertheless, also in this region, an increase in fractal dimension towards the region edges, characteristic of a mobility edge, can be observed.

Finally, for the defected system, we observe that all but one eigenmode retain a fractal dimension close to the one observed in the periodic structure. The one defect eigenmode displays a much higher fractal dimension, due to its much stronger localisation. 

\section{Concluding remarks} \label{sec:conclusion}

In this paper, we have introduced a variety of tools and concepts to understand subwavelength localisation in disordered systems and discovered a variety of analogues and contrasts to periodic subwavelength systems.

By formulating the Thouless criterion in this setting and employing a propagation matrix approach, we have introduced natural generalisations of the notions of bandgap to this setting and given a complete characterisation of the resulting spectral regions based solely on the constituent block properties by proving a Saxon-Hutner type theorem in this setting. 
We have also introduced a defect mode prediction based on the discrete Green function of the disordered system of resonators, together with a decay rate prediction based again on block properties. Furthermore, by investigating the behaviour of block disordered systems under global perturbations, we have observed that these systems behave much analogously to their periodic counterparts. Finally, we have shown that both the landscape function and the fractal dimension of eigenvectors of the capacitance matrix can be applied to detect, respectively, localised eigenmodes and mobility edges in disordered structures even in the presence of strongly localised hybridised bound eigenmodes. 

On the one hand, our results in this paper open the door to the study of wave localisation in other classes of disordered structures such as quasiperiodic or hyperuniform ones (see, for instance, \cite{brynquasi,hyperuniform}) as well as in non-Hermitian systems of subwavelength resonators with gauge potentials \cite{ammari.barandun.ea2024Mathematical,ammari.barandun.ea2023NonHermitian,yao.wang2018Edge,hatano.nelson1996Localization,yokomizo.yoda.ea2022NonHermitian,rivero.feng.ea2022Imaginary,jahan}. On the other hand, our findings here motivate the study of the distributions of the eigenvalues and the localisation properties of the eigenvectors of products of two matrices; one is a diagonal matrix with independent randomly perturbed entries and the other is an $M$-matrix. The study of random M-matrices with correlated entries would also be of great interest in wave physics in the subwavelength regime. In general, little is known about the statistical properties of eigenvalues and eigenvectors of random matrices whose entries are correlated. 

\section*{Acknowledgments}
  The work of AU was supported by the Swiss National Science Foundation grant number 200021--200307. The authors thank Bryn Davies, Erik Hiltunen, Antti Knowles, Ping Liu, and Svitlana Mayboroda for insightful discussions. 

\section*{Code availability}
The software used to produce the numerical results in this work is openly available at \\ \href{https://doi.org/10.5281/zenodo.15373013}{https://doi.org/10.5281/zenodo.15373013}.

\appendix

\section{Thouless criterion for banded $M$-matrices} \label{appendixA}
{

    In this appendix, we show that the methodology developed in this paper is not particular to matrices which are capacitance matrices of systems of subwavelength resonators. Consider an arbitrary banded Hermitian $M$-matrix $A\in \C^{N\times N}$, that is, a matrix with off-diagonal entries less than or equal to zero and nonnegative eigenvalues and denote its bands $\bm b_{-K}, \dots, \bm b_K$:
    \begin{gather*}
        A = \left(\begin{array}{ccccc}
             \bm b_0^{(1)}& \cdots & \bm b_K^{(1)}& \\
             \vdots & \ddots &  & \ddots & \\
             \bm b_{-K}^{(1)}&&\ddots&&\bm b_K^{(N-K)}  \\
             & \ddots&&\ddots&\vdots \\
             &&\bm b_{-K}^{(N-K)} & \cdots & b_{0}^{(N)}
        \end{array}\right).
    \end{gather*}
    
    We can then impose quasiperiodic boundary conditions on $A$ to get $A^\alpha \in \C^{N\times N}$ 
    \begin{gather*}
        A^\alpha = \left(\begin{array}{ccccccccc}
             \bm b_0^{(1)}& \cdots & \bm b_K^{(1)}&0&\cdots&0&\bm e^{-\i\alpha} b_{-K}^{(N-K+1)} & \cdots & e^{-\i\alpha}\bm b_{-1}^{(N)} \\
             \vdots & \ddots &  & \ddots &&&& \ddots & \vdots \\
             \bm b_{-K}^{(1)}&&\ddots&&\ddots&&&&  e^{-\i\alpha}\bm b_{-K}^{(N)} \\
             0&\ddots&&\ddots&&\ddots&&&0\\
             \vdots&&\ddots&&\ddots&&\ddots&&\vdots\\
             0&&&\ddots&&\ddots&&\ddots&0\\
             e^{\i\alpha}\bm b_K^{(N-K+1)}&&&&\ddots&&\ddots&&\bm b_K^{(N-K)}\\
             \vdots& \ddots&&&&\ddots&&\ddots&\vdots \\
             e^{\i\alpha}\bm b_1^{(N)}&\cdots& e^{\i\alpha}\bm b_K^{(N)}&0&\cdots&0&\bm b_{-K}^{(N-K)} & \cdots & \bm b_{0}^{(N)}
        \end{array}\right)
    \end{gather*}
    for $\alpha\in [-\pi,\pi]$\footnote{As the notion of system length $\iL$ is not generally well defined for a general $M$-matrix we have essentially chosen $\iL=1$.}.
    Note that this construction assumes that the bands are always known in full, \emph{i.e.}, $\bm b_k \in \C^N$ for all $k\in -K, \dots, K$ instead of $\bm b_k \in \C^{N-K}$. This is the $K$-banded analogue of having to choose the spacing $s_N$ when periodising the capacitance matrix in \cref{eq:qpcapmat}. In some situations, there are natural ways of obtaining the full bands, for example if the band entries are periodic or sampled from some random distributions. 
    However, even if this is not the case, there is a way to overcome. By truncating $A\in \C^{N\times N}$ to $A' \in \C^{(N-K)\times (N-K)}$, we obtain all the band information necessary to periodise $A'$ and obtain $(A')^\alpha$. If the size of $A$ is significantly larger than the number of bands $N \gg K$, then this introduces only a small error.

    Having periodised $A^\alpha$, we can calculate the band functions\footnote{Note that this is overloaded terminology and the $N$ band functions are not related to the $2K+1$ diagonal bands $\bm b_{-K},\dots, \bm b_K$ of $A$.} $\lambda_1(\alpha), \dots, \lambda_N(\alpha)$ and their corresponding level shifts $\delta\lambda_i$.
    Since the mean level spacings $\Delta(\lambda_i)$ can be estimated in the same way as for the capacitance matrix, we can calculate the Thouless ratio $g(\lambda_i)$ as in \cref{eq:thoulessratio}. 
    
    This allows the use of the Thouless criterion of localisation \cref{eq:thoulesscriterion} for banded Hermitian $M$-matrices. In \cref{fig:M-matrix variation}, we can see that the observed band function variations behave analogously to the capacitance matrix case in \cref{fig:band_function_qpbc}.
}

\begin{figure}[h]
    \centering
    \includegraphics[width=\textwidth]{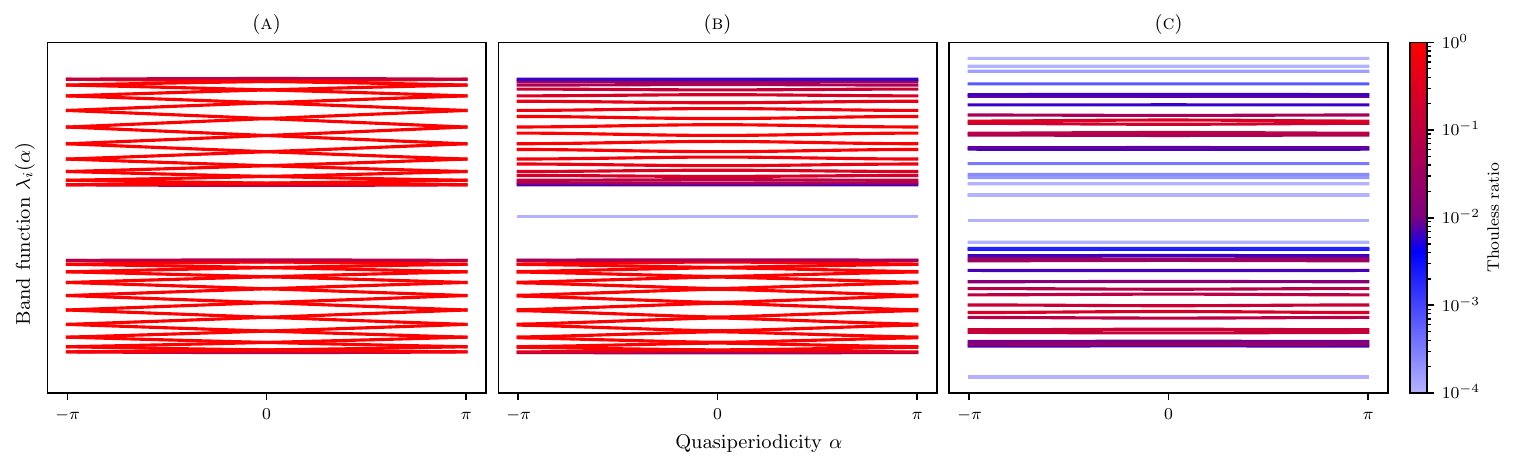}
    \caption{{ The band functions when imposing quasiperiodic boundary conditions on a (\textsc{a}) periodic, (\textsc{b}) defected and (\textsc{c}) randomly perturbed banded Hermitian $M$-matrix. As in \cref{fig:band_function_qpbc}, for each band the colour indicates its variation. 
    }} 
    \label{fig:M-matrix variation}
\end{figure}

\printbibliography
\end{document}